\pdfoutput=1 
\documentclass[dvipsnames]{siamart251216}
\usepackage{amsmath,amssymb}
\usepackage{array}
\DeclareMathOperator*{\argmin}{argmin}
\usepackage{subcaption}
\usepackage[T1]{fontenc}
\usepackage[utf8]{inputenc} 
\usepackage{verbatim}
\usepackage{fancyvrb}
\definecolor{linkcolor}{rgb}{0.0,0.3,0.5}
\newsiamthm{claim}{Claim}
\newsiamremark{remark}{Remark}
\newsiamremark{hypothesis}{Hypothesis}
\crefname{hypothesis}{Hypothesis}{Hypotheses}
\usepackage[all]{hypcap}
\usepackage{graphicx}
\usepackage{mathrsfs}
\usepackage{xspace}
\usepackage{amssymb}
\usepackage[normalem]{ulem} 
\usepackage{bm}             
\usepackage{microtype}      
\usepackage[english]{babel}
\usepackage{blindtext}
\DeclareMathAlphabet{\mathpzc}{OT1}{pzc}{m}{it}
\graphicspath{
  {figs/} 
}
\usepackage{pgfplotstable}
\usepackage{booktabs}
\usepackage{pgfplots}
\pgfplotsset{compat=1.18}
\newcommand{\psiphys}{\Psi}
\newcommand{\psimode}{\Psi}
\newcommand{\cfactor}{C}

\title{Radiation outer boundary conditions and\\
near-to-far field signal transformations\\
for the Bardeen-Press equation}
\author{
Som Dev Bishoyi\thanks{Department of Mathematics,
Center for Scientific Computing
and Data Science Research,
University of Massachusetts,
Dartmouth, MA 02747, USA
(\email{sbishoyi@umassd.edu}).}
\and Scott E. Field\thanks{Department of Mathematics,
Center for Scientific Computing
and Data Science Research,
University of Massachusetts,
Dartmouth, MA 02747, USA
(\email{sfield@umassd.edu}).}
\and
Stephen R. Lau\thanks{Department of
Mathematics and Statistics,
University of New Mexico,
Albuquerque, NM 87131, USA
(\email{lau@unm.edu}).}
}

\begin{document}
\maketitle
\begin{abstract}
Several theoretical and astrophysical problems
---including gravitational-wave modeling for extreme
mass-ratio inspirals--- require accurate time-domain
solutions of the spin-weight $s=-2$ Teukolsky equation in
Boyer-Lindquist coordinates. Because such simulations are
performed on finite computational domains, they typically
introduce an artificial outer boundary where nontrivial
boundary conditions must be imposed. If these conditions
are inaccurate, then spurious reflections and slowly-growing
unphysical modes may corrupt long-time evolutions. We
develop and implement exact radiation outer boundary
conditions for the Bardeen-Press equation (a 
harmonic moment of the $a=0$ Teukolsky equation), 
making the artificial boundary transparent at any finite
radius. We also construct near-to-far field teleportation
kernels that map field data recorded at finite radius $r_1$
to the data reaching $r_2 > r_1$. The possible choice
$r_2 = \infty$ corresponds to asymptotic waveform evaluation,
that is propagation of the data to future null infinity. We
show that both boundary and teleportation kernels are well
approximated by exponential sums, with associated error
bounds. Implemented in a time-domain solver, our kernel-based
boundary conditions eliminate unphysical late-time growth
and give the correct late-time decay rates, affording 
efficient long-duration simulations for waveform modeling
and related blackhole perturbation calculations.
\end{abstract}

\section{Introduction}
\label{Sec:Introduction}

This paper addresses exact radiation outer boundary conditions
(ROBC) and near-to-far field transformations\footnote{
We describe such transformations as radial {\em teleportation},
a memorable term but a slight misnomer.}
for blackhole perturbations in the Teukolsky formalism
with Kerr angular momentum parameter $a=0$.
The Teukolsky master equation describes scalar, vector,
and tensor perturbations of Kerr blackholes~\cite{Teukolsky:1973ApJ}.
Numerical solution of the Teukolsky equation is a primary
means of studying black hole perturbations with the following
applications:
quasi-normal mode spectra~\cite{Berti:2009kk,Leaver:1985ax,
Berti:2025hly};
gravitational waveforms~\cite{Islam:2024vro,Sundararajan:2007jg,
Hughes:2019zmt,Lim:2022veo,Wardell:2021fyy,Rink:2024swg};
extremal blackhole physics~\cite{Burko:2020wzq,
Bishoyi:2024lqm, Gelles:2025gxi, Aretakis:2023ast};
and
black-hole scattering theory and superradiant
amplification~\cite{Teukolsky:1974yv, Brito:2015oca},
including radiation reaction and environmental
effects~\cite{Dolan:2012yt,Cardoso:2021wlq}.
In Boyer-Lindquist coordinates the homogeneous equation is
\begin{align}\label{eq:teuk0}
\begin{split}
\Big[& \frac{(r^2 + a^2)^2}{\triangle}
- a^2\sin^2\theta\Big]\partial^2_t \psiphys
+ \frac{4Mar}{\triangle}
\partial_{t}\partial_{\phi}\psiphys
+ \Big[\frac{a^2}{\triangle}
- \frac{1}{\sin^2\theta}\Big]
\partial^2_\phi\psiphys
\\
& -\triangle^{-s} \partial_r
(\triangle^{s+1}\partial_r \psiphys)
- \frac{1}{\sin\theta}\partial_\theta
(\sin\theta \partial_\theta\psiphys)
- 2s\Big[\frac{a(r-M)}{\triangle}
+ \frac{\mathrm{i}\cos\theta}{\sin^2\theta}\Big]
\partial_\phi\psiphys
\\
& + 2s\Big[r + \mathrm{i}a\cos\theta
- \frac{M(r^2-a^2)}{\triangle}\Big]\partial_t\psiphys
+ (s^2\cot^2\theta - s)\psiphys = 0,
\end{split}
\end{align}
where $s$ is the spin weight of $\psiphys$, and $\triangle = r^2-2Mr+a^2$
in terms of the blackhole mass $M$ and angular momentum per unit mass $a$.
Setting $a=0$, the Teukolsky equation reduces to the Bardeen-Press
equation~\cite{Bardeen:1973xb} governing curvature perturbations of the
Schwarzschild geometry. This paper considers the $a=0$ case (zero
Kerr angular momentum), while retaining all possible values
$s={\pm 2, \pm 1, 0}$ for the field's spin. The $s=\pm 2$ equations describe
the radiative degrees of freedom of the gravitational field: at future null
infinity, for example, the two gravitational-wave polarizations can be
reconstructed from $\psiphys$-fields for $s=-2$.

Beyond their theoretical interest, radiation outer boundary conditions and
near-to-far field transformations are essential for numerical simulations.
Long-time evolutions ideally stem from
domain reduction. Such reduction entails outer boundary conditions
which determine a computed solution agreeing with the open-problem solution
restricted to the finite computational domain. The open Cauchy problem is
posed on a spatially infinite domain. The far-field signal reaching
future null infinity is of particular practical importance, since gravitational-wave
detectors are well modeled as idealized observers located there. Accordingly,
computing the waveform at future null infinity is often the primary
purpose of a simulation. 

The subject of radiation boundary conditions is an established and ongoing
research front in wave phenomena, with developments so manifold that a
comprehensive review is impossible. Instead, we focus on selected works
from the relativity literature. Although each generalizes corresponding
procedures for the ordinary wave equation, the gauge issue associated
with the Einstein equation is absent in that setting. This
complication amplifies the difficulty of moving from scalar to tensor
fields and from linear to nonlinear theory. The works we mention are
therefore not mere extensions of their counterparts for the ordinary wave
equation.

Buchman and Sarbach \cite{Buchman_2006} described radiation boundary
conditions for the linearized Einstein equations, generalizing work on the
ordinary wave equation by Bayliss and
Turkel~\cite{https://doi.org/10.1002/cpa.3160330603}. Their approach
assumes a spherical outer boundary and introduces a hierarchy of local
differential conditions for the Weyl scalar $\Psi_4$ based on the angular
index $\ell$. Subsequent refinements \cite{Buchman_2007} allowed for a
non-spherical boundary and introduced nonlocality to capture backscatter.
Reference \cite{Rinne:2008vn} developed absorbing boundary conditions
based on the wave-equation approach of Hagstrom and
Hariharan~\cite{HAGSTROM1998403}. By introducing auxiliary boundary
variables, this method avoids the numerically problematic high-order
differential conditions associated with a Bayliss-Turkel-style hierarchy.
Although the conditions involve formal infinite expansions, a closure
condition fixes the approximation order and defines incoming
characteristic information. As before, this work goes well beyond the wave
equation in addressing gauge issues. Moreover, the approach has been
implemented within the {\tt SpEC} infrastructure~\cite{SpECwebsite} and
tested \cite{Buchman_2024}. Their method also contends with issues that
would arise in extending our convolution-based approach to the Einstein
equations.

Other strategies for outer boundary conditions and
far-field signal recovery include the hyperboloidal-layer method
\cite{Zenginoglu:2010cq, Zenginoglu:2011zz} pioneered by Zenginoglu and,
more recently, Cauchy-characteristic matching
(CCM)~\cite{Ma:2024hzq,Ma:2023qjn}. Hyperboloidal layers, and related
approaches based on hyperboloidal
foliations~\cite{Ansorg:2016ztf, PanossoMacedo:2019npm}, use radial
compactification together with a suitable time transformation to adjoin
(part of) future null infinity as a boundary of the computational domain.
Their physical appeal is obvious, and they avoid the spurious
high-frequency reflections present in older spatial compactification
schemes, although other work \cite{APPELO20094200} has also addressed
this issue. Similar in spirit, CCM couples a Cauchy evolution in the
near zone to a characteristic evolution in the wave zone with a
matching interface. Through compactification, the characteristic evolution
reaches future null infinity, thereby avoiding artificial outer boundary
conditions.

This article builds on the
approach~\cite{alpert2000rapid,alpert2002nonreflecting} of Alpert,
Greengard, and Hagstrom (AGH) for ROBC in flat spacetime. Starting from
the spatially and temporally nonlocal boundary conditions specifying exact
domain reduction to a finite spherical region, AGH developed a framework
based on rational approximation for rapid implementation, with
theoretically established accuracy and complexity estimates. References
\cite{Lau:2004as,Lau:2004jn,Lau:2005ti} generalized the AGH approach
(without some of its theoretical foundations) to the Regge-Wheeler and
Zerilli equations governing metric perturbations of Schwarzschild black
holes, while near-to-far field transformations in the same setting were
addressed in \cite{Benedict:2012kw}. Ironically, the investigation
\cite{Benedict:2012kw} for black-hole perturbations appeared before the
corresponding theoretical and numerical work for the ordinary wave
equation in \cite{Greengard_2014,Field:2014cka}. Reference
\cite{Greengard_2014} established that high-$\ell$ near-to-far field
teleportation is ill-conditioned, but numerical investigations
\cite{Field:2014cka} suggest that our implementation ameliorates this
ill-conditioning and remains useful for $\ell$ beyond 100, a large value
for relativity applications. Historically, Wilcox
\cite{wilcox1959initial} and Tokita \cite{Tokita1972} were the first to
address near-to-far field transformations.

A key difficulty with solutions to the Bardeen-Press equation stems from
the following issue raised by Hughes \cite{PhysRevD.62.044029}. Consider
the $s=-2$ case and waves of a given given frequency $\omega$ propagating
near an outer boundary at $r=r_b \gg 1$. In terms of the tortoise
coordinate $r_*$ (see below), the outgoing and incoming solutions in the
vicinity of the boundary have the asymptotic forms $e^{i\omega r_*}$ and
$r^{-4}e^{-i\omega r_*}$, ignoring the overall $e^{-i\omega t}$ temporal
dependence; see \Cref{app:ConfluentHeun}. The general solution will then be
{\em asymptotic} to $C_1 e^{i\omega r_*} + C_2r^{-4}e^{-i\omega r_*}$, with
$C_1$ and $C_2$ determined both by the initial data and the boundary
conditions. Ideally, the boundary conditions enforce $C_2=0$. In practice,
due to roundoff error, we can only expect $C_2r_b^{-4}$ close to machine
precision $\varepsilon_\mathrm{mach}$. Therefore, $C_2\simeq r_b^4
\varepsilon_\mathrm{mach}$ can be appreciably larger than
$\varepsilon_\mathrm{mach}$, meaning roundoff errors incurred at
the boundary (indeed, at any large radii) will affect in the
interior where $r^{-4}C_2 \simeq (r_b/r)^4 \varepsilon_\mathrm{mach} 
\gg \varepsilon_\mathrm{mach}$. This description highlights the 
importance of ROBC which afford faithful numerical evolutions on
short domains.

The new results in this paper are the following. First, we
theoretically analyze exact ROBC and near-to-far field teleportation
for the Bardeen Press equation, identifying the correct variable
transformations needed to define the corresponding convolution kernels.
While this analysis is similar to corresponding derivations for the
Regge-Wheeler and Zerilli equations, it proves more difficult.
Second, we show that both radiation and teleportation kernels for the
Bardeen-Press equation admit accurate rational approximations in the
frequency domain, corresponding to sum-of-exponential approximations
in the time-domain. Again, while these results are similar to earlier
work on metric gravitational perturbations, the kernels in the
Bardeen-Press setting prove harder to approximate. We also present
new error bounds for the compressed kernel approximations.
Finally, via long-time, time-domain evolutions, we demonstrate the
effectiveness of kernel-based ROBC and teleportation in the present
setting. With both implementations and simulations on finite radial
intervals, we compute the correct late-time tail behavior of
radiation reaching future null infinity.

This paper is organized as follows.
\Cref{sec:radiation_outer_boundary_conditions} analyzes ROBC for
the Bardeen-Press equation, introducing the relevant kernels.
\Cref{sec:asymptotic_waveform_evaluation_and_signal_teleportation}
describes the kernels for near-to-far field signal teleportation.
\Cref{sec:kernel_eval_comp} studies rational approximation of kernels.
\Cref{sec:NumExp} presents time-domain
evolutions studying the effectiveness of ROBC and teleportation. Several
appendices supplement the main text.

\section{Radiation outer boundary conditions}
\label{sec:radiation_outer_boundary_conditions}

We follow \cite{Lau:2004as,Lau:2004jn,Lau:2005ti} which 
derived ROBC for the Regge-Wheeler and Zerilli
equations. Chandrasekhar discovered \cite{10.1098/rspa.1975.0066,Chandrasekhar:BHs}
that solutions to the Regge-Wheeler
equation can be related to solutions to the Bardeen-Press equation via
an eponymous transformation. However, we have been unable to
exploit the Chandrasekhar transformation to obtain
ROBC for the Bardeen-Press equation from those associated
with the Regge-Wheeler equation. We shall instead apply the
procedures outlined in \cite{Lau:2004as,Lau:2004jn} to the Bardeen-Press
equation \cref{eq:teuk0-1p1-v4}. The strategy is
to (i) derive an ODE in the Laplace
frequency domain which describes the non-trivial piece of the outgoing
solution, (ii) formulate the exact frequency-domain
ROBC in terms of this non-trivial piece, and (iii) approximate the
ROBC in a fashion that carries over to the time-domain.

\subsection{Bardeen-Press equation}
\label{Subsec:Evolution}
Derivation of the Bardeen-Press equation from \cref{eq:teuk0} with $a=0$
starts with expansion of the Teukolsky master function:
\cite{Stein:2012ffl,Dolan:2012yt,Barack:2017oir}
\begin{align}
\label{eq:anstaz}
\psiphys(t,r,\theta,\phi) =
\frac{1}{r \triangle^s} \sum_{\ell=|s|}^{\infty}
\sum_{m=-\ell}^{\ell}\psimode_{\ell m}(t,r)
{}_sY_{\ell m}(\theta, \phi).
\end{align}
This expansion features the spin-weighted spherical harmonics
${}_sY_{\ell m}(\theta,\phi)$ and accounts for large-$r$ behavior. For
the harmonics we adopt the same definition and conventions as
\cite{Thorne:1980ru}. Substitution of \eqref{eq:anstaz} into
\eqref{eq:teuk0} and subsequent use of well-known identities for the
spin-weighted harmonics (notably Eq.~(25) from~\cite{Barack:2017oir})
reduces the $a=0$ Teukolsky equation to a decoupled set of 1+1
wave equations for the mode coefficients $\psimode_{\ell m}(t,r)$.
These equations take a simple form when expressed in terms of the
tortoise coordinate $r_{*}$ defined by $d r_{*}=(r^{2}/ \triangle) d r$.
In $(t,r_{*})$ coordinates the equation obeyed by
$\psimode_{\ell m}$ is
\begin{align}\label{eq:teuk0-1p1-v3}
  \partial^2_t\psimode_{\ell m}
- \partial^2_{r_*}\psimode_{\ell m}
+ \frac{2s(r-3M)}{r^2}\partial_t\psimode_{\ell m}
+ \frac{2s(r-M)}{r^2}\partial_{r_*}\psimode_{\ell m}
+ V(r)\psimode_{\ell m} = 0,
\end{align}
where the potential is
\begin{align}\label{eq:potential}
V(r) = \left(1-\frac{2M}{r}\right)\left[
\frac{2M(s+1)}{r^{3}}+ \frac{(\ell - s)(\ell+s+1)}{r^2}\right].
\end{align}
Note that \eqref{eq:teuk0-1p1-v3} is identical to the 
($s=-2$ and $a=0$) mode equation used
in~\cite{Barack:2017oir}, although that work uses Eddington-Finkelstein
null coordinates. Reference \cite{Barack:2017oir}
has already analyzed the ingoing (toward the horizon) and outgoing
(toward future null infinity) solutions to this equation.
For a wave of fixed frequency $\omega$ these special solutions behave as
\begin{align}\label{eq:BarackSpecialSolutions}
\psimode_{\ell}^\mathrm{down}(t,r)
\stackrel{r\rightarrow 2M^+}{\sim}
\mathrm{e}^{-i \omega (t + r_{*})},\qquad
\psimode_{\ell}^\mathrm{out}(t,r)
\stackrel{r\rightarrow \infty}{\sim}
\mathrm{e}^{-i \omega (t - r_{*})}.
\end{align}
Notably, here neither the downgoing nor outgoing wave amplitudes
vary with $r$. To verify correctness of the behavior
\cref{eq:BarackSpecialSolutions}, notice that the wave operator
in \cref{eq:teuk0-1p1-v3} tends to
$\partial_t^2-\partial_{r_*}^2
-\frac{1}{2}sM^{-1}(\partial_t-\partial_{r_*})$ and
$\partial_t^2-\partial_{r_*}^2$ as $r\rightarrow 2M^{+}$
and $r\rightarrow\infty$, respectively.

We rewrite \eqref{eq:teuk0-1p1-v3} in terms of dimensionless coordinates
$\tau = t/(2M)$, $\rho = r/(2M)$, and $\rho_* = \rho + \log(\rho -1)$,
finding
\begin{align}\label{eq:teuk0-1p1-v4}
\begin{split}
\partial^2_\tau\psimode_{\ell m}
-\partial^2_{\rho_*}\psimode_{\ell m}
+\frac{2s(\rho-\frac{3}{2})}{\rho^2}
 \partial_\tau\psimode_{\ell m}
+\frac{2s(\rho-\frac{1}{2})}{\rho^2}
 \partial_{\rho_*}\psimode_{\ell m}
+V(\rho)\psimode_{\ell m}
= 0.
\end{split}
\end{align}
Here we have redefined both the meaning of $\psimode_{\ell m}$ and the
potential
\begin{align}\label{eq:potential_rho}
V(\rho) = \left(1-\frac{1}{\rho}\right)
\left[\frac{(s+1)}{\rho^{3}}+
\frac{(\ell - s)(\ell+s+1)}{\rho^2}\right].
\end{align}

\subsection{Solution structure in Laplace frequency domain}
\label{Sec:LaplaceDomain}
Formal Laplace transformation\footnote{Except for
\cref{sub:l2_norm_bounds}, the variable $s$ exclusively
represents spin and {\em not} the Laplace-frequency dual
to time $t$. We use $\sigma$ as the dimensionless
Laplace-frequency dual to dimensionless time $\tau$.}
of \cref{eq:teuk0-1p1-v4} yields\footnote{With the
transformation $\widehat{\Psi}_\ell =
[\rho(\rho-1)]^{s/2}\widehat{\Theta}_\ell$,
the ODE \cref{eq:Teuk_laplace0} takes the
Schr\"{o}dinger form
$$
-\frac{d^2 \widehat{\Theta}_\ell}{d\rho_*^2} +
(\sigma^2 + U)\widehat{\Theta}_\ell = 0,
\qquad
U = 2s\sigma \frac{(\rho-\frac{3}{2})}{\rho^2}
+ \frac{s^2(\rho-\frac{1}{2})^2 + s(\rho-1)^2}{\rho^4} + V(\rho).
$$
Note that $U(\rho)$ falls-off as $1/\rho$.
}
\begin{align}
\label{eq:Teuk_laplace0}
-\frac{d^2\widehat{\psimode}_{\ell m}}{d\rho_*^2}
+\frac{2s (\rho-\frac{1}{2})}{\rho^2}
 \frac{d\widehat{\psimode}_{\ell m}}{d \rho_*}
+\frac{2s\sigma (\rho-\frac{3}{2})}{\rho^2}
 \widehat{\psimode}_{\ell m}
+V(\rho)\widehat{\psimode}_{\ell m}
+\sigma^2\widehat{\psimode}_{\ell m} = 0,
\end{align}
where $\widehat{\psimode}_{\ell m}(\sigma,\rho)
= (\pounds\psimode_{\ell m}(\cdot,\rho))(\sigma)$ is the
Laplace transform of $\psimode_{\ell m}(\tau,\rho)$.
Given a boundary radius $\rho_b$, we
assume that the initial data $\{\psimode_{\ell m}(0,\rho),
(\partial_\tau\psimode_{\ell m})(0,\rho)\}$
vanishes for $\rho > \rho_b$. The {\em homogeneous}
character of \cref{eq:Teuk_laplace0} stems from this
assumption. With
\begin{align}\label{eq:CHtrans}
\widehat{\psimode}_{\ell m}
= e^{-\sigma\rho_*}\Big(\frac{\rho-1}{\rho}\Big)^s
\widehat{\Phi}_{\ell m},
\end{align}
\eqref{eq:Teuk_laplace0} becomes
\begin{align}\label{eq:BardeenPress-Phihat}
\frac{d^2\widehat{\Phi}_{\ell m}}{d\rho^2}
+P(\rho, \sigma)\frac{d\widehat{\Phi}_{\ell m}}{d\rho}
+Q(\rho)\widehat{\Phi}_{\ell m} = 0,
\end{align}
where for later use we have defined
\begin{align}\label{eq:PandQ2}
\begin{split}
P(\rho, \sigma) & =
-2\sigma - \frac{3s+1}{\rho} + \frac{1+s-2\sigma}{\rho-1}
\\
Q(\rho) & =
\frac{(s+1)(2s+1)}{\rho^2}
-\frac{(s+1)^2 + \ell (\ell + 1)}{\rho (\rho-1)}.
\end{split}
\end{align}
\Cref{eq:BardeenPress-Phihat} is a
normal form of the confluent Heun equation;
see \cref{app:ConfluentHeun}.

As shown in \cref{app:ConfluentHeun}, about $\rho=\infty$
(the point at infinity), we may choose asymptotically
outgoing and incoming solutions to
\eqref{eq:BardeenPress-Phihat} with the expansions
\begin{align}\label{eq:Thomeoutinc}
\widehat{\Phi}_{\ell}^\text{out}(\sigma,\rho)
\sim 1 + \frac{L}{2\sigma\rho},
\qquad
\widehat{\Phi}_{\ell}^\text{inc}(\sigma,\rho)
\sim \rho^{2s}e^{2\sigma\rho_*},
\end{align}
where $L=\ell(\ell+1)-s(s+1)$. The special solution
$\Psi_\ell^\text{out}$ in \cref{eq:BarackSpecialSolutions}
stems from inverse Laplace transform of
$(\rho-1)^s \rho^{-s} e^{-\sigma\rho_*}
\widehat{\Phi}_{\ell m}^\text{out}$ (and passage back
to physical coordinates).

Following \cite{Lau:2004jn,Lau:2004as} with $z = \sigma \rho$,
use\footnote{
Along with other expressions which follow, $W_{\ell}(z;\sigma)
= {}_sW_{\ell}(z;\sigma)$ depends on the spin weight $s$ appearing in
\cref{eq:BardeenPress-Phihat}. We suppress this dependence.}
$W_{\ell}(z;\sigma)$ for the outgoing solution to
\eqref{eq:BardeenPress-Phihat}, so
\begin{align}
W_{\ell}(\sigma \rho;\sigma) = 
\widehat{\Phi}^\text{out}(\rho,\sigma).
\end{align}
This choice of argument structure facilitates description of
our procedures for numerical evaluation of the outgoing solution
and radiation kernel; see~\cref{app:KernelEvaluation}. 
Moreover, for $s=0$ and under the $\sigma \rightarrow 0$ limit
with $z$ held fixed, the outgoing solution is related to
MacDonald's function~\cite{AbramowitzStegun}
\begin{align}
K_{\ell+1/2}(z) = \sqrt{\frac{\pi}{2z}}e^{-z}
W_\ell(z;0).
\end{align}

\subsection{Radiation kernels}
\label{subsec:ROBC}
Write the outgoing solution to \eqref{eq:Teuk_laplace0} as
\begin{align}
\label{eq:outgoing_FD}
\widehat{\psimode}_{\ell}^\mathrm{out}(\sigma,\rho) =
e^{-\sigma\rho_*}
\Big(\frac{\rho-1}{\rho}\Big)^s W_\ell(\sigma\rho;\sigma),
\end{align}
and assume initial data compactly supported for $\rho < \rho_b$.
Then the solution to \eqref{eq:Teuk_laplace0} is
\begin{align}\label{eq:apsiout}
\widehat{\psimode}_{\ell m}(\sigma,\rho) =
a_{\ell m}(\sigma)
e^{-\sigma\rho_*}
\Big(\frac{\rho-1}{\rho}\Big)^s W_\ell(\sigma\rho;\sigma),
\end{align}
for $\rho \geq \rho_b$,
where $a_{\ell m}(\sigma)$ depends on the initial data and
possible source terms. We now apply the Sommerfeld operator
$\sigma + \partial_{\rho_*}$ to \eqref{eq:apsiout}, thereby finding
\begin{align}\label{eq:Eq2}
\begin{split}
(\sigma + \partial_{\rho_*})\widehat{\psimode}_{\ell m}(\sigma,\rho)
& = \frac{\rho-1}{\rho^2}
\widehat{\omega}_\ell(\sigma,\rho)
\widehat{\psimode}_{\ell m}(\sigma,\rho)
+ \frac{s}{\rho^2}
\widehat{\psimode}_{\ell m}(\sigma,\rho),
\end{split}
\end{align}
where we have defined the frequency-domain radiation kernel (FDRK)
\begin{align}
\label{eq:ROBC_kernelFD}
\widehat{\omega}_\ell(\sigma,\rho) = 
\frac{\rho \partial_\rho W_{\ell}(\sigma \rho;\sigma)}{
W_{\ell}(\sigma \rho;\sigma)} =
\sigma \rho \frac{W_{\ell}'(\sigma \rho;\sigma)}{
W_{\ell}(\sigma \rho;\sigma)},
\end{align}
with the prime indicating differentiation in the argument $z$.
In \cref{eq:Eq2} $\widehat{\omega}_\ell$ depends on $\ell$, $s$,
and $\rho$, whereas the solution $\widehat{\psimode}_{\ell m}$
has additional dependence on $m$ through (generally $m$-dependent)
sources and initial data.

We argue that the kernel $\widehat{\omega}_\ell$
decays for large $\sigma$, and thus has a classical inverse Laplace
transform. The first expansion in \eqref{eq:Thomeoutinc} suggests
that $\widehat{\omega}_\ell$ will decay as $\sigma\rightarrow \infty$.
To confirm this expectation, notice that the kernel obeys the
Riccati-type equation
\begin{align}\label{eq:Riccatirho}
\frac{d\widehat{\omega}_\ell}{d\rho}
+ \frac{\widehat{\omega}_\ell^2}{\rho}
+ \Big[P(\rho, \sigma) - \frac{1}{\rho}\Big] \widehat{\omega}_\ell
+ \rho Q(\rho) = 0,
\end{align}
where $P(\rho, \sigma)$ and $Q(\rho)$ are given in \cref{eq:PandQ2}.
Consider then the formal expansions
\begin{align}\label{eq:expandomega}
\widehat{\omega}_\ell & = \frac{\widehat{\omega}_\ell^{(-1)}}{\sigma}
                        + \frac{\widehat{\omega}_\ell^{(-2)}}{\sigma^2}
                        + O\Big(\frac{1}{\sigma^3}\Big),
\qquad
P(\rho, \sigma)\,\widehat{\omega}_\ell
= -\frac{2\rho}{\rho-1}\widehat{\omega}_\ell^{(-1)}
+ O\Big(\frac{1}{\sigma}\Big).
\end{align}
Substitution of these expansions into \eqref{eq:Riccatirho} followed by
balance of $O(1)$ terms yields
\begin{align}
\widehat{\omega}_\ell^{(-1)}
= \frac{(s+1)(2s+1)(\rho-1)}{2\rho^2}
-\frac{(s+1)^2 + \ell(\ell + 1)}{2\rho}.
\end{align}
Balance at $O(1/\sigma)$ determines $\widehat{\omega}_\ell^{(-2)}$
in terms of $\widehat{\omega}_\ell^{(-1)}$ and it $\rho$ derivative.
The consistency of this formal calculation suggests that the kernel
$\widehat{\omega}_\ell$ decays for large $\sigma$.

With evaluation at $\rho=\rho_b$ and inverse Laplace
transformation, the frequency-domain ROBC \cref{eq:Eq2}
corresponds to the nonlocal time-domain boundary condition
\begin{align}\label{eq:ROBC_kernelTD}
\begin{split}
\Big(
\frac{\partial \psimode_{\ell m}}{\partial \tau}
+ \frac{\partial \psimode_{\ell m}}{\partial \rho_*}
\Big)\Big|_{\rho=\rho_b}
& = \frac{\rho_b-1}{\rho_b^2}
\int_0^{\tau} \omega_{\ell}(\tau - \tau',\rho_b)
\psimode_{\ell m}(\tau',\rho_b)d\tau'
+ \frac{s}{\rho_b^2}\psimode_{\ell m}(\tau,\rho_b).
\end{split}
\end{align}
The first term on the right-hand side is a convolution between
the solution $\psimode_{\ell m}(\tau,\rho_b)$ and the inverse Laplace
transform $\omega_{\ell}(\tau;\rho_b)$ of the kernel
\cref{eq:ROBC_kernelFD}. The boundary condition
\cref{eq:ROBC_kernelTD} is a key result. It determines exact
reduction of the infinite radial domain, introducing no spurious
reflections from the boundary, while still allowing for physical
backscatter. If $s=\ell=M=0$,
then the right-hand side of \cref{eq:ROBC_kernelTD} is zero and
the equation reduces to the familiar Sommerfeld condition.  

\section{Signal teleportation and asymptotic waveform evaluation}
\label{sec:asymptotic_waveform_evaluation_and_signal_teleportation}
This section describes radial teleportation: given the history
$\psimode_{\ell m}(\tau,\rho_1)$ of the solution at $\rho_1\leq\rho_b$,
we wish to reconstruct the outgoing signal
$\psimode_{\ell m}(\tau,\rho_2)$ at any larger radius $\rho_2 >
\rho_1$. When $\rho_2=\infty$, we refer to the reconstruction
as {\em asymptotic waveform evaluation} (AWE).
In electromagnetics the same idea is called near-to-far field
transformation~\cite{TafloveHagness2005,LiTafloveBackman2005,Muller2011},
emphasizing that data on a finite-radius surface can be mapped to the
radiative signal at larger radii. In numerical relativity, closely
related approaches are commonly referred to as waveform extrapolation,
where one fits the signal in powers of $1/\rho$ (in practice $1/r$) and
then extrapolates to $\rho_2=\infty$ (ideally to null infinity).
However, unlike waveform extrapolation, our procedure is exact for
solutions to \cref{eq:teuk0}, given initial data supported within
$\rho_1$. We follow the framework developed previously for the
flatspace~\cite{Field:2014cka} and
Regge-Wheeler/Zerilli~\cite{Benedict:2012kw} wave equations.

\subsection{Teleportation kernels}
The form \cref{eq:outgoing_FD} of the outgoing Laplace-domain
solution implies that
\begin{align}
\label{eq:FD_teleportation}
e^{\sigma(\rho_2^* - \rho_1^*)}\widehat{\psimode}_{\ell m}(\sigma, \rho_2)
& = \cfactor(\rho_1,\rho_2)\widehat{\phi}_{\ell}(\sigma,\rho_1,\rho_2)
\widehat{\psimode}_{\ell m}(\sigma,\rho_1)
+ \cfactor(\rho_1,\rho_2)\widehat{\psimode}_{\ell m}(\sigma,\rho_1),
\end{align}
where we have defined the frequency-domain teleportation kernel
\begin{align}\label{eq:phihatDEFN}
\widehat{\phi}_{\ell}(\sigma,\rho_1,\rho_2)
= -1 + \frac{W_\ell(\sigma\rho_2;\sigma)}{W_\ell(\sigma\rho_1;\sigma)},
\end{align}
as a $\sigma$-dependent function with $\rho_1$ and $\rho_2$ held fixed, and
\begin{align}
\cfactor(\rho_1,\rho_2) =
\left[\frac{\rho_1(\rho_2-1)}{\rho_2(\rho_1-1)}\right]^s.
\end{align}
The $-1$ in \eqref{eq:phihatDEFN} ensures that
$\widehat{\phi}_\ell(\sigma,\rho_1,\rho_2)$ decays as $|\sigma|\to\infty$;
its inverse Laplace transform is then a classical function.
Heuristically, the kernel 
$\widehat{\phi}_\ell(\sigma,\rho_1,\rho_2)$ captures the nontrivial part
of the propagation from $\rho_1$ to $\rho_2$, namely, the
modification of the wave beyond the retarded-time shift. Indeed,
when $\ell=s=0=M$ (so the potential $V=0$ and $\rho=\rho_*$),
the kernel vanishes, and \cref{eq:FD_teleportation} reduces to
$e^{\sigma(\rho_2 - \rho_1)}\hat{\psimode}_{00}(\sigma, \rho_2)
=\hat{\psimode}_{00}(\sigma, \rho_1)$. The corresponding time-domain
formula is then pure time-translation associated with
finite-speed propagation,
\begin{align}
\psimode_{00}(\tau + (\rho_2-\rho_1),\rho_2) =
\psimode_{00}(\tau,\rho_1).
\end{align}
With $\rho_2\to\infty$, the formula
\cref{eq:FD_teleportation} yields an exact evaluation of the
asymptotic waveform. 

To numerically compute the teleportation kernel,
we begin with the relationship \cite{Benedict:2012kw}
between the teleportation and radiation kernels:
\begin{align}\label{eq:TLPfromRBC}
\widehat{\phi}_\ell(\sigma,\rho_1,\rho_2) &  = -1 +
\underbrace{\exp\left[
\int_{\rho_1}^{\rho_2}\frac{\widehat{\omega}_\ell(\sigma, \eta)}{\eta}
d\eta\right]}_{W_\ell(\sigma\rho_2;\sigma)/W_\ell(\sigma\rho_1;\sigma)}.
\end{align}
With infrastructure for accurate numerical evaluation of
radiation kernels in place, numerical evaluation of
$\widehat{\phi}_\ell(\sigma,\rho_1,\rho_2)$ boils down to
accurate quadrature for the integral appearing in
\cref{eq:TLPfromRBC}. The methods we have used are
described in~\cite{Field:2014cka}.

Inverse Laplace transform of the frequency-domain
relation \eqref{eq:FD_teleportation} yields the time-domain
teleportation formula:
\begin{align}
\label{eq:teleportation_TD}
\begin{split}
\psimode_{\ell m}( & \tau + (\rho_2^* - \rho_1^*), \rho_2)
\\
& = \cfactor(\rho_1,\rho_2)
\int_0^{\tau} \phi_{\ell}(\tau - \tau',\rho_1, \rho_2)
\psimode_{\ell m}(\tau',\rho_1)d\tau'
+ \cfactor(\rho_1,\rho_2) \psimode_{\ell m}(\tau, \rho_1) \,.
\end{split}
\end{align}
Unlike ROBC which are imposed throughout an evolution,
signal teleportation \eqref{eq:teleportation_TD} can be
performed either on the fly or as post-processing,
provided $\psi_{\ell m}(\tau,\rho_1)$ has been suitably
recorded. Teleportation kernels
have been used to study wave phenomenology,
such as eccentricity-excited tails~\cite{Islam:2024vro}.

\section{Kernel profiles and compression}
\label{sec:kernel_eval_comp}
This section studies frequency-domain radiation kernels (FDRKs).
It exhibits and surveys FDRK profiles, and it sketches the process
by which we compress an FDRK. This section touches on similar
issues for teleportation kernels. Our work here relies on 
schemes for {\em numerical} evaluation of
$\widehat{\omega}_{\ell}(iy,\rho_b)$ along the imaginary axis
$\sigma=iy$, the inversion contour for the inverse Laplace
transform. These schemes are described both in
\cref{app:KernelEvaluation} and, more fully,
\cite{Lau:2004jn}. Through these schemes, we are able to
evaluate $\widehat{\omega}_{\ell}(iy,\rho_b)$ with near
double-precision machine accuracy, and better when using quad
arithmetic. Therefore, from the standpoint of this section,
we view $\widehat{\omega}_{\ell}(iy,\rho_b)$ as an exact
evaluation.

\subsection{Profile survey}
\label{subsec:ROBCProfiles}
Since $\ell = 2$ is the dominant mode in spherical-harmonic
decomposition of gravitational wave signals, we first consider
$\ell=2$ for different values of spin-weight $s$. When $s=0$,
the Bardeen-Press equation coincides with the spin-zero
Regge-Wheeler equation. \Cref{fig:differentspins} shows the real
and imaginary FDRK profiles for $\ell=2,\rho_b=60$ and $s=-2,0,2$.
Similar plots for spin-two Regge-Wheeler and Zerilli FDRKs appear
in \cite{Lau:2004jn,Lau:2004as}. The $s=-2$ FDRK exhibits sharper
features than seen in either the $s=0,2$ cases or the spin-two
Regge-Wheeler and Zerilli kernels (qualitatively similar to the
depicted $s=0$ kernel). In addition, for $s=-2$ the plots suggest
the existence of special frequencies $iy \simeq \pm i0.02$ at
which the FDRK $\widehat{\omega}_2(iy,60)$ is purely imaginary,
a feature not present in the Regge-Wheeler and Zerilli settings.
\begin{figure}[!htbp]
    \centering
    \includegraphics[width=0.75\columnwidth]{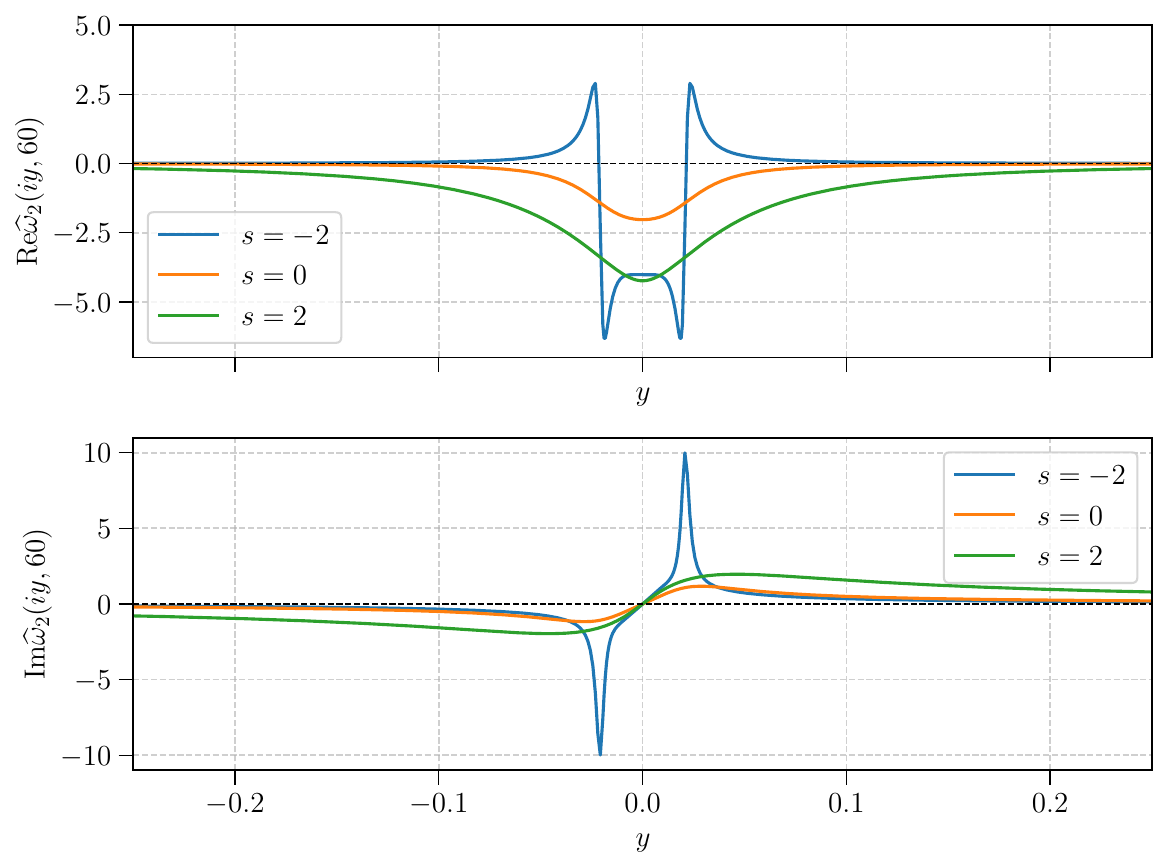}
    \caption{Profiles for the Bardeen-Press
    radiation kernel $\widehat{\omega}_2(iy,60)$
    for different spin-weights $s=-2, 0, 2$.
    The $s=2$ kernel has been scaled up by a factor
    of $100$ for the sake of comparison.}
    \label{fig:differentspins}
\end{figure}

Higher $(\ell,m)$ modes allow for accurate modeling of
gravitational-wave signals, especially in the self-force 
framework and, more generally, in black-hole perturbation
theory \cite{LISAConsortiumWaveformWorkingGroup:2023arg}.
\Cref{fig:ROBC_vs_sigma} shows the real and imaginary
profiles for an $\ell=64,\rho_b=60$ radiation kernel and
$s=-2,0,2$. While at first sight, these kernel profiles appear
qualitatively similar, important differences are evident.
For example, beyond a neighborhood of the origin, the real
part of the $s=-2$ kernel appears positive for all sufficiently
large $y$, unlike the other two cases. This feature
is also evident for $\ell=2$.

We have observed that the imaginary-axis profiles for an
$s=-2$ Bardeen-Press FDRK are markedly different than those
associated with the wave equations for metric perturbations.
Based on this observation, we hope to examine the analytic
structure of the outgoing solution $W_\ell(\sigma\rho_b;\sigma)$
in the left-half $\sigma$-plane. Similar examinations in
the Regge-Wheeler and Zerilli scenarios appear in \cite{Lau:2005ti}.
In that reference analytic representation of an FDRK involves
both a simple-pole sum and a branch-cut contribution.
\begin{figure}[!htbp]
    \centering
    \includegraphics[width=0.75\columnwidth]{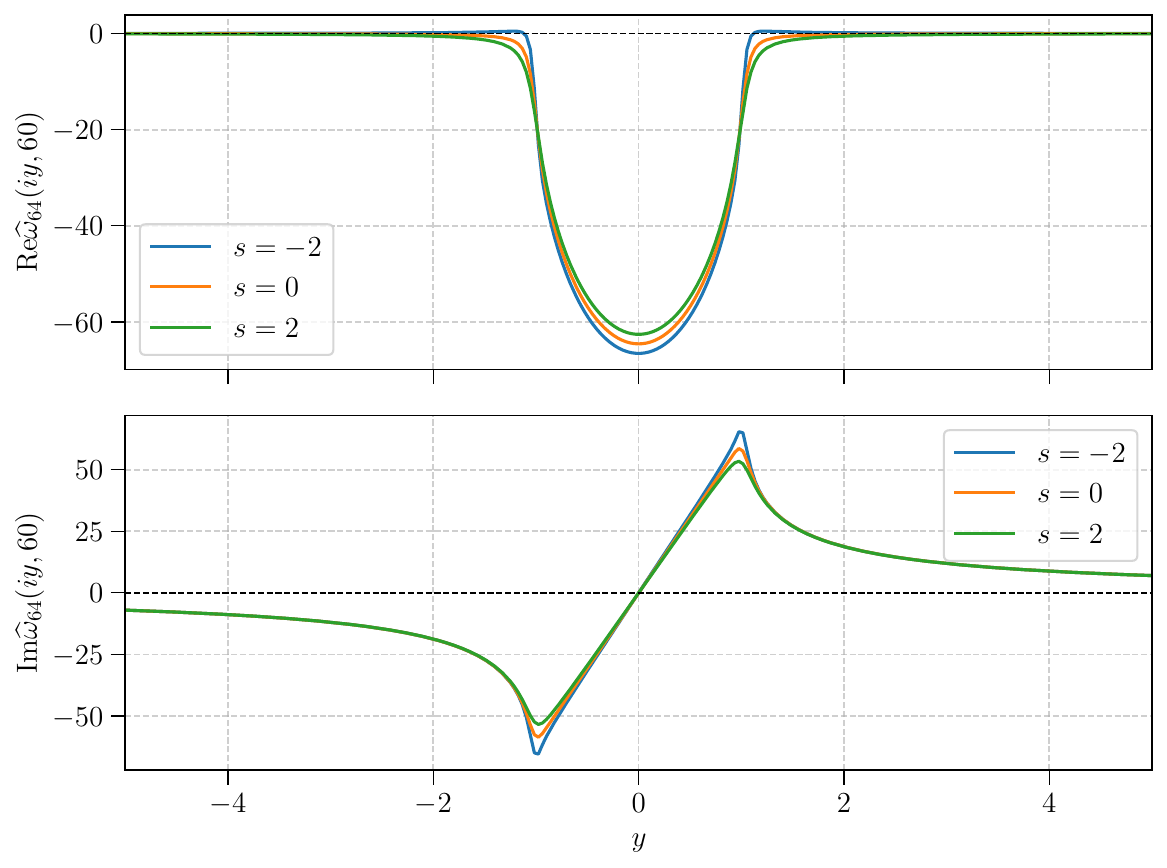}
    \caption{Bardeen-Press radiation kernel
    $\widehat{\omega}_{64}(iy,60)$
    for $s=-2, 0, 2$.}
    \label{fig:ROBC_vs_sigma}
\end{figure}

\subsection{Kernel compression: approximation by rational functions}
Here we summarize results from
\cite{alpert2000rapid,alpert2002nonreflecting,Lau:2004jn,xu2013bootstrap,BIZZOZERO2017118}.
\label{subsec:SumofPoles}
\subsubsection{Description}
The goal is rational approximation of a FDRK,
\begin{align}\label{eq:sumofpoles}
\widehat{\omega}_{\ell}(\sigma,\rho_b)\simeq
\widehat{\xi}_{\ell}(\sigma,\rho_b)
\equiv \sum_{k=1}^{d} \frac{\gamma_{k}}{\sigma-\beta_{k}} \,.
\end{align}
In practice, the pole locations $\beta_{k}$ for the approximation
lie in the left half-plane, as required for stability.
The complex numbers $\{\beta_{k},\gamma_{k} \}_{k=1}^d$
depend on the real and imaginary kernel profiles
(depicted in \cref{subsec:ROBCProfiles}); whence they inherit dependence
on $s$, $\ell$, and $\rho_b$, but we suppress this dependence. The
approximation should obey the uniform relative error bound
\begin{align}
\label{eq:SumofPoles_error}
\sup_{y\in\mathbb{R}}
\frac{\left|
\widehat{\omega}_{\ell}(iy,\rho_b)-
\widehat{\xi}_{\ell}(iy,\rho_b)
\right|}
{\left|\widehat{\omega}_{\ell}(iy,\rho_b)\right|}
< \varepsilon,
\end{align}
where $\varepsilon$ is a prescribed tolerance. This bound
guarantees that $\widehat{\xi}_{\ell}(\sigma,\rho_b)$
closely matches $\widehat{\omega}_{\ell}(\sigma,\rho_b)$
along the imaginary axis, which in turn
yields accurate approximation of the time-domain kernel
\begin{align}\label{eq:xiSOE}
\omega_{\ell}(\tau,\rho_b)\simeq \xi_{\ell}(\tau,\rho_b)
= \sum_{k=1}^{d} \gamma_{k}e^{\beta_{k}\tau}.
\end{align}
This connection is made precise by the error bounds
in~\cref{sub:l2_norm_bounds}.

Likewise, we approximate the freqency-domain
teleportation kernel \cref{eq:TLPfromRBC} as
\begin{align}\label{eq:xiHAT_T}
\widehat{\phi}_\ell(\sigma,\rho_1,\rho_2) \simeq
\widehat{\xi}_\ell^T(\sigma,\rho_1,\rho_2)
= \sum_{q=1}^d \frac{\gamma_q^T}{\sigma - \beta_q^T},
\end{align}
where the teleportation parameters
$\{\gamma_q^T, \beta_q^T\}_{q=1}^d$
also depend on $\ell$, $s $, $\rho_1$, $\rho_2$. This
dependence is suppressed throughout. Inverse Laplace
transform of \cref{eq:xiHAT_T} yields
\begin{align}\label{eq:xi_T}
\phi_\ell^T(\tau,\rho_1,\rho_2) \simeq
\xi_\ell^T(\tau,\rho_1,\rho_2) 
= \sum_{q=1}^d \gamma_q^T e^{\beta_q^T \tau},
\end{align}
an accurate approximation of the time-domain
teleportation kernel.
\begin{figure}[!htbp]
\centering
\includegraphics[width=0.75\columnwidth]{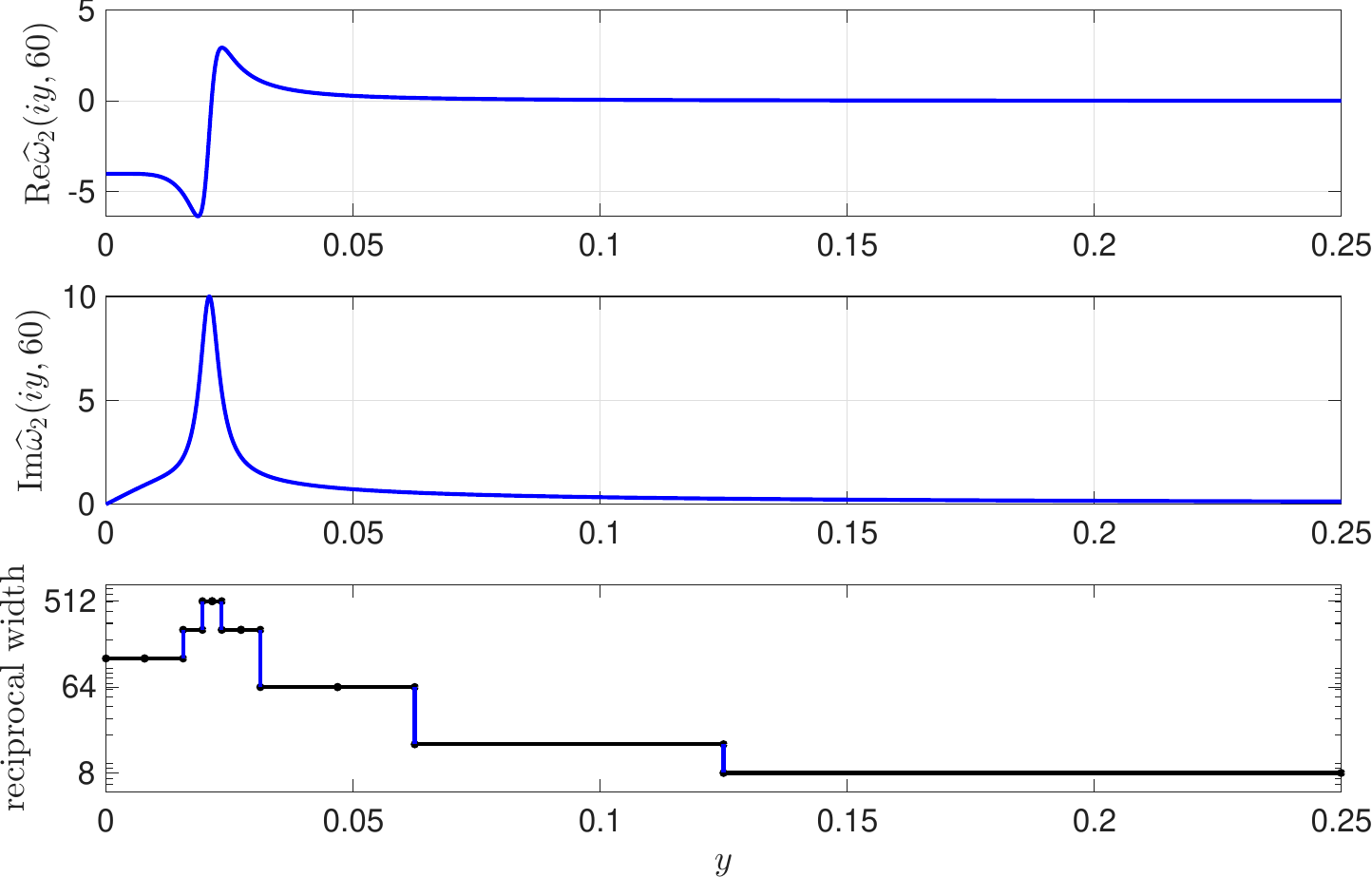}
\caption{
Binary-tree construction of Xu and Jiang. The bottom panel
shows $N_\mathrm{int}=10$ subintervals adapted to the variation
of $\widehat{\omega}_2(iy,60)$ for $s=-2$, with the reciprocal
subinterval width as the ordinate.
}
\label{fig:LevelComparison}
\end{figure}

\subsubsection{Compression algorithm}
\label{subsec:CompressionAlg}
Introduce a window $[-y_\mathrm{max},y_\mathrm{max}]$,
typically with $y_\mathrm{max}$ in the range $10$ to $100$. We 
partition this window into $2N_\mathrm{int}$ subintervals, with
$Q+1$ uniformly spaced points on each subinterval. We either choose
the partition ``by-hand'' to have more resolution at the origin or
use the binary-tree construction 
\cite{xu2013bootstrap,BIZZOZERO2017118} of Xu and Jiang.
With the Xu-Jiang
construction, variation of the kernel profiles is uniform on the
tree leafs. \Cref{fig:LevelComparison} shows a Xu-Jiang binary
tree adapted to the profiles for $\widehat{\omega}_2(iy,60)$, with
$y_\mathrm{max} = 0.25$ corresponding to the plots in
\cref{fig:differentspins}. The figure shows
$N_\mathrm{int}=10$ subintervals and the profiles are uniformly
smooth on each one. By ``smooth'' we mean that Cheyshev-polynomial
expansions of the profiles on each subinterval obey a convergence
criterion; see \cite{BIZZOZERO2017118} for details.

The resulting adaptive grid $\{y_j\}_{j=0}^{2Q N_\mathrm{int}}$
and corresponding trapezoidal integration weights
$\{\mu_j\}_{j=0}^{2Q N_\mathrm{int}}$ then define
the minimization problem
\begin{align}\label{fig:minProblem}
\argmin_{\{\beta_q,\gamma_q\}_{q=1}^d}
\sum_{j=0}^{2Q N_\mathrm{int}}
\mu_j
\Big|\widehat{\omega}_\ell(iy_j,\rho_b) -
\sum_{k=1}^d \frac{\gamma_k}{iy_j-\beta_k}\Big|^2,
\end{align}
where the $\widehat{\omega}_\ell(iy_j,\rho_b)$ here is really
our numerical evaluation at $y_j$ of the true kernel. In practice,
the number $d$ of poles is increased until the set
$\{\beta_q,\gamma_q\}_{q=1}^d$ satisfies
a suitably discretized version of \cref{eq:SumofPoles_error}
for given fixed $\varepsilon$. The process \cite{alpert2000rapid}
for solving the minimization problem \eqref{fig:minProblem} has
been rightly described as an ingenious use of the Gram-Schmidt
algorithm. See \cite{xu2013bootstrap,BIZZOZERO2017118} for
further refinements.

\begin{figure}[!htbp]
    \centering
     \includegraphics[width=0.75\columnwidth]{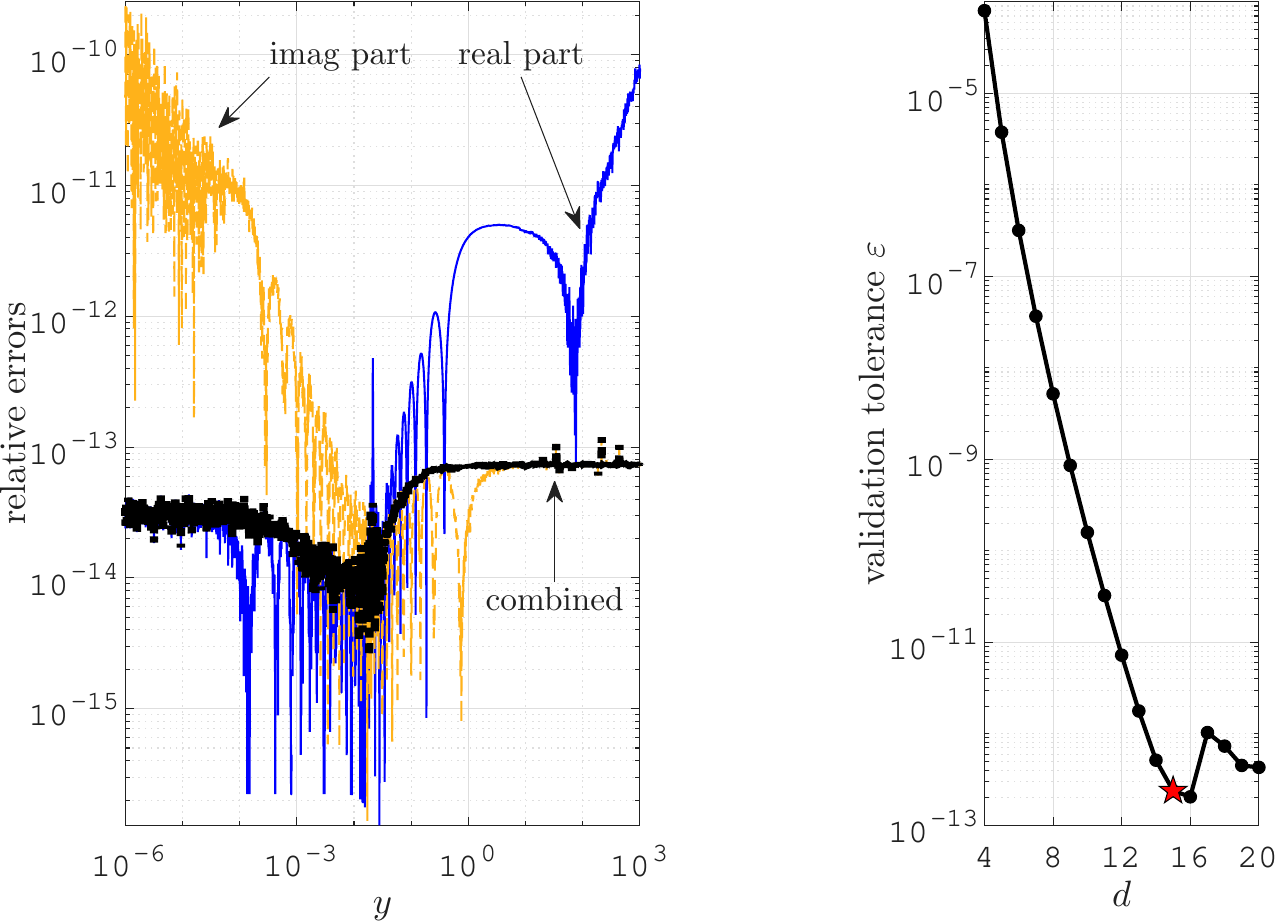}
    \caption{Left panel: Relative errors between the FDRK
             $\widehat{\omega}_2(iy,60)$ and its sum-of-poles
             approximation $\widehat{\xi}_2(iy,60)$. The thin blue curve is
             relative error in the real part, the thin dashed gold curve is the
             relative error in the imaginary part, and the solid black curve is
             the combined relative error; see the text.
             Right panel:  Relative errors $\varepsilon$ as computed by
             validation as a function of the number $d$ of poles in the
             approximation. The left panel corresponds to the
             red star in the right panel. For the star $\frac{1}{2}
             \varepsilon \simeq 1.19\times 10^{-13}$, the maximum relative
             (combined) error in the left panel.
    }
    \label{fig:Errcompression}
\end{figure}

\subsubsection{Accuracy validation of compressed kernels}
We have just described rational approximation of
$\widehat{\omega}_\ell(\sigma,\rho_b)$ via the Alpert, Greengard,
Hagstrom compression algorithm.
While the algorithm relies on an adaptive $y$-grid and associated
quadrature rule to produce the rational approximation, 
validation of \cref{eq:SumofPoles_error} is performed on a
auxiliary dense reference grid. This validation, an accuracy
assessment of the sum-of-poles approximation \cref{eq:sumofpoles},
compares the rational approximation against independent numerical
evaluation of $\widehat{\omega}_{\ell}(iy;\rho_b)$ on the imaginary axis. 
Here, numerical evaluation of both $\widehat{\omega}_{\ell}(iy,\rho_b)$
and $\widehat{\xi}_{\ell}(iy,\rho_b)$ has been carried out in
double precision.

\Cref{fig:Errcompression} shows the relative errors
between the Bardeen-Press
FDRK with $s=-2$, $\ell=2$, and $\rho_b=60$ and its $d=15$ rational
approximation. The comparison grid (with $y$-points labelled $Y_\alpha$)
used for the figure is described in the next paragraph.
Away from the origin, both the real and imaginary
parts of the kernel become small, but the combined relative error
levels off below $10^{-12}$. We interpret this as a consequence of
measuring relative error where the underlying quantity is close to
zero. Near the origin, the same effect appears for the relative error
of the imaginary part, where the kernel's imaginary part (denominator)
vanishes, but not for the real part. We have choosen the number of
poles $d=15$ so that the maximum relative compression error defined in
\cref{eq:SumofPoles_error} is below $10^{-12}$. By the error bound
\cref{eq:l2_estimate_rel_T}, this ensures that the time-domain
convolution error from the sum-of-poles approximation remains small. 

The right panel of \cref{fig:Errcompression} shows convergence of our
rational approximations as $d$ increases, with the $d=15$ case, shown as
a red star, corresponding to left-panel plot described in the last paragraph.
To obtain the right-panel plot, we have taken the
following steps: (i) generation of numerical profile values
$\{\widehat{\omega}_2(iy_j,60)\}_{j=0}^{2QN_\mathrm{int}}$ on a
$y$-grid determined by $y_\mathrm{max}=100$ and an $N_\mathrm{int}=19, Q=64$
binary tree; (ii) AGH compression of $d=4,5,\dots,20$ pole-sums,
each approximating the generated numerical FDRK; (iii) validation of
each pole-sum to estimate its $\varepsilon$ tolerance.
For validation, we evaluate the relative error on an auxiliary
$y$-grid $\{Y_\alpha\}_{\alpha=0}^{16384}$, where $Y_0=0$ and
$Y_1,\ldots,Y_{16384}$ are logarithmically spaced from $10^{-6}$ to
$10^3$. In comparison with \cref{eq:SumofPoles_error}, we then estimate
the tolerance by taking $\varepsilon$ to be twice the maximum sampled
relative error, with the extra factor of $2$ included as a safety margin.

\section{Implementation and time-domain numerical experiments}
\label{sec:NumExp}
This section tests implementation of radiation and teleportation
kernels in time-domain evolutions of the $s=-2$ Bardeen-Press
equation \cref{eq:teuk0-1p1-v3}. Our evolutions of
\cref{eq:teuk0-1p1-v3} stem from a first-order reduction
based on characteristic variables
$W_{\ell m}^{\pm} = -\partial_t \psimode_{\ell m} \pm \partial_{r_*} \psimode_{\ell m}$.
In terms of these variables the evolution equations are
\begin{align}\label{eq:1+1D}
\begin{split}
\partial_t \psimode_{\ell m} 
& = -\frac{1}{2}\Big(W^+_{\ell m} + W^-_{\ell m}\Big)\\
\partial_t W^+_{\ell m}
& = -\partial_{r_*}W^+_{\ell m}
    -\frac{2s\left(r-2M\right)}{r^2} W^-_{\ell m}
    +\frac{2sM}{r^2}W^+_{\ell m}
    + V(r)\psimode_{\ell m}\\
\partial_t W^-_{\ell m}
& = +\partial_{r_*}W^-_{\ell m}
    -\frac{2s\left(r-2M\right)}{r^2} W^-_{\ell m}
    +\frac{2sM}{r^2}W^+_{\ell m}
    + V(r)\psimode_{\ell m}.
\end{split}
\end{align}
While we continue to view the fields as functions
of $r$, for example as $\psimode_{\ell m}(t,r)$,
evolution of \cref{eq:1+1D} uses the tortoise coordinate.
The system \cref{eq:1+1D} indicates that $W^+$ propagates
left-to-right and $W^-$ right-to-left. With the
computational domain taken as the interval $[a,b]$,
boundary conditions are then needed for $W^+$ at $r_* = a$
and $W^-$ at $r_* = b$. To complete the system \cref{eq:1+1D},
we must specify $W^+(t, r_a)$ and $W^-(t, r_b)$.

Spatial discretization of \cref{eq:1+1D}
has been carried out both with a nodal discontinous Galerkin
(DG) method based on Legendre polynomials and a nodal Chebyshev
method. See \cite{Field:2009kk,Vishal:2023fye} for details on
the DG implementation. While the plots shown below have been
generated with the Chebyshev solver, all have been cross-checked
with the DG solver. We use classic fourth-order Runge-Kutta as
the time-stepper.

\subsection{Physical time-domain radiation and
teleportation kernels}
\label{subsec:evol-ROBC}
Assuming initial data supported compactly within the boundary
$r_* = b$, the exact outgoing boundary condition is the Laplace
convolution [cf.~\cref{eq:ROBC_kernelTD}]
\begin{align}\label{eq:lapconv}
W^-_{\ell m}(t,r_b) = -\frac{f(r_b)}{r_b}
\int_{0}^t \Omega_\ell(t - t', r_b)
\psimode_{\ell m}(t', r_b)dt' -
\frac{2Ms}{r_b^2} \psimode_{\ell m}(t, r_b),
\end{align}
where $f(r) = 1- 2M/r$ and the physical time-domain boundary
kernel is
\begin{align}
\Omega_\ell(t,r_b) = \frac{1}{2M}\omega_\ell(t/(2M),r_b/(2M)).
\end{align}
The time-domain kernel $\omega_\ell(\tau,\rho)$ on the right-hand
side also appears in \cref{eq:xiSOE}. The boundary condition
\eqref{eq:lapconv} is approximated via replacement of
$\Omega_\ell(t,r)$ by
\begin{align}
\Xi_\ell(t, r_b) = \frac{1}{2M}\xi_\ell(t/(2M),r_b/(2M))
= \sum_{q=1}^d \frac{\gamma_{q}}{2M}
e^{\frac{1}{2}\beta_q t/M},
\end{align}
with the last expression following from \cref{eq:xiSOE}.
As mentioned, the set $\{\gamma_q, \beta_q\}_{q=1}^{d}$
depends on $\ell$, $s$, and $\rho_b=r_b/(2M)$, but this dependence has
been suppressed.

With $\Xi_q(t,r_b) = \exp(\frac{1}{2}\beta_q t/M)$, the
constituent convolutions $(\Xi_q * \psimode_{\ell m})(t)$ can be
evaluated through integration of the ODE
\begin{align}
\frac{d}{dt} (\Xi_q * \psimode_{\ell m})(t) =
\frac{\beta_q}{2M}(\Xi_q * \psimode_{\ell m})(t)
+ \psimode_{\ell m}(t,r_b),\qquad
q=1,\dots,d.
\end{align}
In effect, the implementation adds $d$ auxiliary variables
$\{(\Xi_q * \psimode_{\ell m})(t): q=1,\dots, d\}$ at the boundary,
evolved along side the system \cref{eq:1+1D}.

Physical teleportation from $r_1$ to $r_2 > r_1$ has
the form [cf.~\cref{eq:teleportation_TD}]
\begin{align}\label{eq:asympwave}
\begin{split}
\psimode_{\ell m}&(t + (r_2^* - r_1^*), r_2)
\\
& = \cfactor(\rho_1, \rho_2)\int_{0}^t
\Phi_{\ell}(t-t', r_1, r_2)
\psimode_{\ell m}(t', r_1)dt' 
  + \cfactor(\rho_1, \rho_2)
\psimode_{\ell m}(t, r_1),
\end{split}
\end{align}
where in terms of the left-hand side of \cref{eq:xi_T}
the physical teleportation kernel is
\begin{align}
\Phi_{\ell}(t, r_1, r_2)
=\frac{1}{2M}\phi_\ell(t/(2M),r_1/(2M),r_2/(2M)).
\end{align}
The convolution \cref{eq:asympwave} is approximated via
replacement of $\Phi_{\ell}(t, r_1, r_2)$ by
\begin{align}
\Xi_\ell^T(t, r_1, r_2) =
\frac{1}{2M}\xi_\ell^T(t/(2M),r_1/(2M),r_2/(2M))
= \sum_{q=1}^d \frac{\gamma_q^T}{2M}e^{\frac{1}{2}\beta_q^T t/M},
\end{align}
with the last equality following from \cref{eq:xi_T}.
As mentioned, the teleportation parameters
$\{\gamma_q^T, \beta_q^T\}_{q=1}^d$
also depend on $\ell$, $s $, $\rho_1$, and $\rho_2$, but we
continue to suppress this dependence. Integration of ODEs
at the boundary affords a simple implementation of
the convolution. Teleportation may be performed as a post-processing
step, since the relevant ODEs are decoupled from the system
\cref{eq:1+1D}.

\subsection{Quasinormal ringing and decay tails}
\label{subsec:QN_and_tails_time_inf}
The following experiment tests our ROBC and signal teleportation
with a long-time evolution. Consider the $\ell=2,s=-2$
Bardeen-Press system \cref{eq:1+1D} with the following
Gaussian initial data:
\begin{align}\label{eq:inidat}
\psimode(0,r) = 0,
\qquad
W^-(0,r) = \frac{
e^{-\frac{1}{2}(r_* - \mu)^2/\sigma^2}}{\sqrt{2\pi \sigma^2}},
\qquad
W^+(0,r) = W^-(0,r),
\end{align}
where $\mu=30M$ and $\sigma=10M$. Notice that the Gaussian profile is
in tortoise coordinate $r_*$ not $r$. Here, $\Psi(t,r) = 
\Psi_{2m}(t,r)$ for any azimuthal index $|m| \leq \ell = 2$, but for
this experiment we suppress all angular indices. We place the inner
boundary at $r_* = a = -150M$ and the outer boundary at $r_b=120M$,
corresponding to $r_* = b \simeq 128.16M$. The spatial
domain is divided into 37 equal-width subintervals, each with
32 Chebyshev nodal points. We impose a Sommerfeld
condition at the inner boundary and a compressed-kernel approximation
of \cref{eq:lapconv} at the outer boundary. \Cref{tab:ROBC60table}
lists the relevant compressed radiation kernel. Evolving
beyond $t=5000M$ with a time-step size
$\Delta t \simeq 4.0097\times 10^{-3}M$, we record
$|\psimode(t,120M)|$ at the outer boundary, with this history
(orange curve) shown in the top panel of
\cref{fig:tailandLPI_near}.  After an initial phase of
quasinormal-ringing lasting until $t \simeq 350M$, the solution
transitions to a Price tail with power-law decay $t^p$. 
To estimate $p$ empirically, we compute a local power index
(LPI) using a first-order finite-difference approximation to
$p = \partial_{\log t}\log(|\psimode(t,120M)|)$.
As seen in the bottom panel of \cref{fig:tailandLPI_near},
for late times the local power index approaches the Price-law
prediction $t^{-(2\ell+3)} = t^{-7}$, before eventual
contamination by round-off error when $t \gtrsim 2200M$.
The dotted, horizontal, red line corresponds to exact $t^{-7}$ decay.
\begin{figure}[!htbp]
    \centering
    \includegraphics[width=0.9\columnwidth]{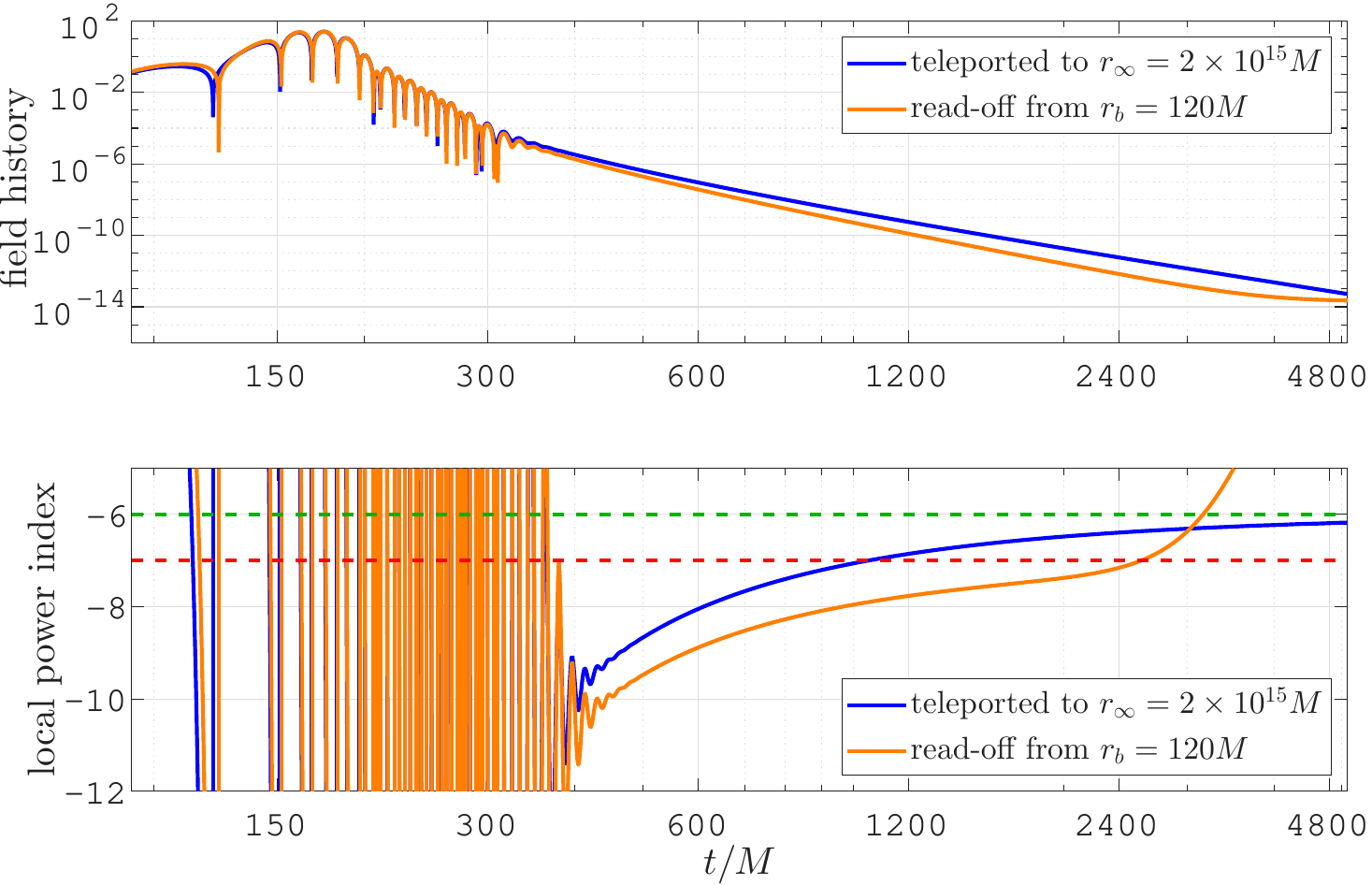}
    \caption{Top panel. Quasinormal ringing and decay tails. Each curve
    corresponds to power law decay. The time shift for the asymptotic
    waveform has not been included. At late times $t\gtrsim 2200$
    the solution is dominated by numerical error.
    Bottom panel. Decay rates for signals
    recorded at $r_1 = r_b = 120M$ and teleported to
    $r_2 = r_\infty = 2\times 10^{15}M$.}
    \label{fig:tailandLPI_near}
\end{figure}

To approximately recover the waveform reaching null infinity
from the simulation described in the last paragraph, we
teleport the history $\psimode(t,r_b=120M)$ to a large radius,
here chosen as $r_\infty = 2\times 10^{15}M$. With $r_1=r_b$
and $r_2 = r_\infty$, this entails
approximation of the convolution formula \cref{eq:asympwave}
to define $\psimode_{\infty}(t) \simeq
\psimode(t + (r_{\infty}^* - b), r_{\infty})$, where we
continue to suppress the angular indices.
The top panel of \Cref{fig:tailandLPI_near} also shows our computed
asymptotic waveform $|\psimode_{\infty}(t)|$ as a time-series (blue curve),
with the corresponding local power index also shown in the bottom panel.
For the teleported signal we have an LPI approaching $-6$ (the horizontal,
dotted, green line), in agreement with the theoretical prediction for
tails at null infinity; see \cite{Price:1972pw,donninger2011proof,Zenginoglu:2008uc}
and table II in~\cite{PhysRevD.100.104025}. This experiment demonstrates that 
a combination of ROBC and signal teleportation captures quasinormal ringing 
and late-time tails at null-infinity.
\begin{figure}[!htbp]
    \centering
    \includegraphics[width=0.85\columnwidth]{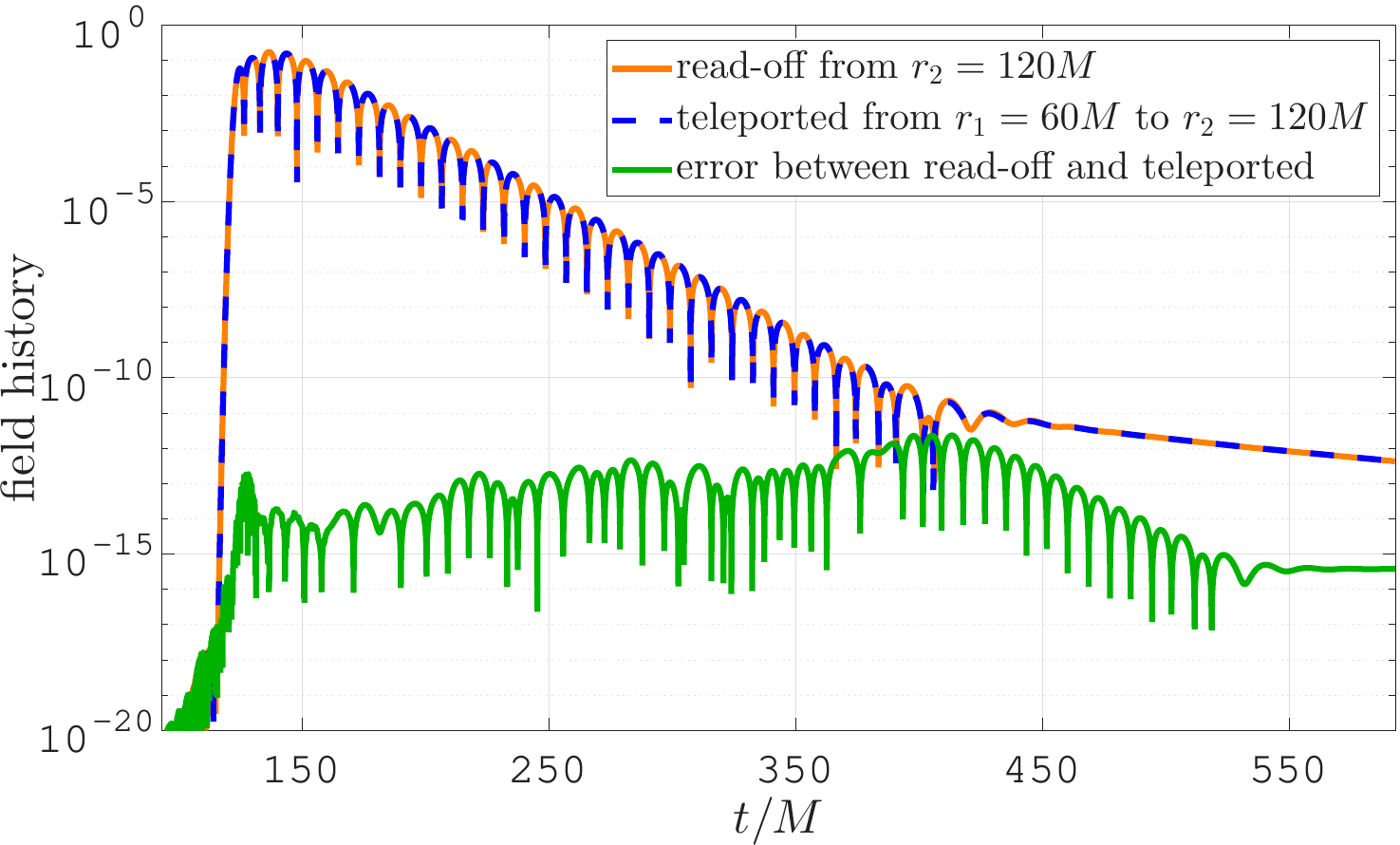}
    \caption{Comparison of the signal
             $|\psimode_{120M}^\mathrm{rec}(t+\Delta r_*)|$
             recorded at $r_2 = 120M$ with time delay (blue)
             and the signal
             $|\psimode^\mathrm{tel}_{60M\rightarrow 120M}(t)|$ teleported
             from $r_1=60M$ to $r_2=120M$ (orange). See the text. The
             absolute error between the signals is also shown (green).}
    \label{fig:teleport_r60tor120}
\end{figure}    

\subsection{Pulse teleportation}
\label{subsec:PulseTeleportation}
This experiment also evolves the $s=-2,\ell=2$ Bardeen-Press system
\eqref{eq:1+1D}, starting with the pulse initial data \cref{eq:inidat},
only now with $\mu=4M$ and $\sigma =M$. We continue to suppress all
angular indices. Using our nodal Chebyshev method, we perform
evolutions for two choices of the domain $[a,b]$. For both choices
$a=-150M$, with the outer boundary
$b=b(r_b)$ fixed by either $r_b=60M$ or $r_b=120M$. We use $37$
subintervals for the short domain $[-150M,b(60M)]$, and $45$ subintervals
for the long domain $[-150M,b(120M)]$. Both evolutions adopt degree-$31$
polynomials on each subinterval and place a Sommerfeld boundary condition
$r_* = a$. We take \cref{eq:lapconv} as the boundary condition at either
$r_b = 60M$ or $r_b = 120M$, with implementation based on the
relevant compressed-kernel from either \cref{tab:ROBC30table}
or \ref{tab:ROBC60table}. For each evolution, we first evolve to
$t_\mathrm{mid} = b-r_*(30M)$, then record the solution over an
observation window of length $t_\mathrm{obs} = 500M$. The observation
window uses $n_\mathrm{steps}=2247672$ time steps, while the initial
evolution uses (the integer closest to) $n_\mathrm{steps}
t_\mathrm{mid}/t_\mathrm{obs}$. This procedure aligns the solution
teleported from $b(60M)$ to $b(120M)$ exactly with the solution recorded at $b(120M)$.

Throughout the observation window on the short domain, we record
the history $\psimode(t,60M)$ as a time series, also numerically teleporting
this series to obtain
\begin{align}
\psimode^{\mathrm{tel}}_{60M\rightarrow 120M}(t) \simeq
\psimode(t + \Delta r_*,120M).
\end{align}
Here, the difference $\Delta r_* = b(120M)-b(60M) \simeq 61.42M$
accounts for the time delay between the outer boundaries of the short and
long domains: $r_*=b(60M) \simeq 66.73M$ and $r_* = b(120M) \simeq 128.16M$,
respectively. This $60M\rightarrow 120M$ teleportation is achieved with
the convolution \cref{eq:asympwave}, taking $r_1=60M$ and $r_2=120M$
and \cref{tab:TLP30to60table} for its approximation. Throughout the
observation window on the long domain, we record the history
$\psimode^{\mathrm{rec}}_{120M}(t) =\psimode(t,120M)$ as a time series.
The experiment is then to compare
$\psimode^{\mathrm{tel}}_{60M\rightarrow 120M}(t)$ with
$\psimode^{\mathrm{rec}}_{120M}(t+\Delta r_*)$. This comparison
is achieved without interpolation through the process described in the last
paragraph, resulting in a time-step size
$\Delta t = t_\mathrm{obs}/n_\mathrm{steps} \simeq 2.2245 \times 10^{-4}M$
over each observation window.
\Cref{fig:teleport_r60tor120} shows that the absolute error 
$|\psimode^{\mathrm{tel}}_{60M\rightarrow 120M}(t)
-\psimode^{\mathrm{rec}}_{120M}(t+\Delta r_*)|$ is orders
of magnitude smaller than $|\psimode^{\mathrm{rec}}_{120M}(t+\Delta r_*)|$.

\section{Conclusion}
\label{sec:conclusion}
We have presented the first time-domain
evolutions of the Bardeen-Press equation with ROBC imposed at a
finite radius. Exact radiation (or transparent) boundary conditions
achieve domain reduction of an open wave problem.
Without exact ROBC, spurious
reflections off the outer boundary contaminate the interior
solution. While these remarks apply broadly to wave problems,
transparent outer boundary conditions are indispensable for
long-time evolutions of the Bardeen-Press equation, a setting
where spurious late-time growth has been reported \cite{Long:2021ufh}
on $\mathcal{O}(1000M)$ timescales. By contrast, solutions to
the Regge-Wheeler equation for perturbations of Schwarzschild exhibit
no such growth. For the Bardeen-Press case, such growth is partly a 
numerical artifact of
large spatial domains, so ROBC are needed to obtain faithful evolutions
on shorter domains.

We have also extended to the Bardeen-Press equation the asymptotic
waveform evaluation procedure developed in \cite{Benedict:2012kw},
allowing recovery of the far-field signal from a time series recorded at
finite radius. We demonstrated that this near-to-far field
transformation can be implemented accurately and simply. For the
dominant $\ell=2,s=-2$ mode, our numerical results recover the expected
late-time decay $t^{-6}$. This approach does not require compactification
or conformal completion. Instead, the far-field signal is recovered in
post-processing through evaluation of the Laplace convolution
\cref{eq:asympwave}. We also demonstrated effective signal teleportation
between two finite radial locations.

The same strategy applies to both nonlocal convolutions considered here.
First, the convolution is written as a product in the Laplace-frequency
domain, which identifies the corresponding frequency-domain kernel.
Next, this kernel is approximated by a sum of simple poles in the
left half-plane, with a prescribed tolerance along the inversion
contour. Inverting the product between the approximate kernel and the
solution then yields a time-domain convolution in terms of damped
exponentials. Depending on the setting, this procedure gives ROBC,
finite-radius teleportation, or asymptotic waveform evaluation. A
similar strategy may also apply to ROBC and teleportation for Kerr
perturbations governed by the Teukolsky equation. In that sense, the
Bardeen-Press case is a first step in that direction.

For the $\ell=2$ cases considered here, implementing transparent
boundary conditions through a kernel convolution costs about as much as
adding only 20 points to the spatial grid. The method requires no
coordinate transformations and applies at any radius beyond the support
of the initial data and sources. The FDRK developed here should also be
useful for enforcing exact outgoing conditions in frequency-domain
solvers, with similar applications to frequency-domain teleportation
kernels. In that setting, one would first apply a Wick rotation before
using an FDRK.

The main remaining challenge is reducing the offline cost of
generating teleportation and AWE tables. 
Kernel generation via \cref{eq:TLPfromRBC} is
expensive \cite{Benedict:2012kw}, often requiring more than $10^{13}$
floating-point operations, and appears to be costlier for
Bardeen-Press than for Regge-Wheeler or Zerilli kernels. For metric
perturbations and a given imaginary frequency $\sigma$, the real and
imaginary quadratures in \cref{eq:TLPfromRBC} involve values of the same
sign, whereas for the $s=-2$ Bardeen-Press equation they appear to
involve both positive and negative values. We are studying alternative
strategies for computing these kernels more efficiently.

\section*{Acknowledgments}
We are grateful for discussions with 
Manas Vishal, Jaryd Domine, Thomas Hagstrom, Gaurav Khanna, Oliver Long,
Anil Zenginoglu, and Anthony Cognilio. We thank Matias Shedden for
rewriting Fortran kernel compression routines in Python.
Some numerical computations were performed on the UMass-URI UNITY
supercomputer supported by the Massachusetts Green High Performance Computing
Center (MGHPCC). We also thank the UNM Center for Advanced Research Computing,
supported in part by the National Science Foundation, for providing the high
performance computing resources used in this work.
The authors acknowledge support of NSF Grant DMS-2309609 and the Office of
Naval Research/Defense University Research Instrumentation Program (ONR/DURIP)
Grant No.~N00014181255.
\appendix

\section{Normal form and monodromy}
\label{app:ConfluentHeun}
The ODE \cref{eq:BardeenPress-Phihat} has the form
\begin{align}\label{eq:Heuninf}
\frac{d^2H}{d\rho^2} +
\Big[
\beta + \frac{\gamma-2\alpha}{\rho}
+ \frac{\delta}{\rho-1}\Big]\frac{dH}{d\rho}
+\Big[
\frac{\alpha(1+\alpha-\gamma)}{\rho^2}
+ \frac{q + \alpha (\beta-\delta)}{\rho(\rho-1)}
\Big] H = 0,
\end{align}
with $\alpha = 1+s$, $\beta = -2\sigma$, $\gamma = 1-s$, $\delta =
1+s-2\sigma$, and $q' \equiv q + \alpha (\beta-\delta) = -\ell(\ell+1)-(s+1)^2$.
For $\sigma\neq 0$ the generalized Riemann scheme
\cite{slavyanov2000special} for \cref{eq:Heuninf} is
\begin{align}\label{eq:generalscheme}
\left[
\begin{array}{cccl}
1               &  1         &   2                    &        \\
0               &  1         & \infty                 & ;\rho  \\
\alpha          &  0         &  0                     & ;q'    \\
1-\gamma+\alpha & 1-\delta   & \gamma+\delta-2\alpha  &        \\
                &            &  0                     &        \\
                &            & -\beta                 &
\end{array}
\right]
=
\left[
\begin{array}{cccl}
1               &  1         &   2                    &        \\
0               &  1         & \infty                 & ;\rho  \\
\lambda_0       & \lambda_1  & \lambda_\infty         & ;q'    \\
\lambda_0'      & \lambda_1' & \lambda_\infty'        &        \\
                &            & \kappa_\infty          &        \\
                &            & \kappa_\infty'         &
\end{array}
\right].
\end{align}
Each of the first three columns of the scheme gives information about
the singular points; the second row is location and the first row
type (1 corresponds to a regular singular point, and 2 to an irregular
singular point of confluent type
\cite{wasow1987asymptotic,erdelyi1956asymptotic}). The remaining
information collects the indicial exponents and Thom\'{e} parameters (see
\cite{erdelyi1956asymptotic} and below).
The factor $q$ (or $q'$) is the {\em accessory parameter}, reflecting
the fact that the equation is only determined by its monodromy up to a free
parameter. Equation \cref{eq:Heuninf} is a normal form of the confluent Heun
equation (with one indicial exponent vanishing at both $1$ and $\infty$),
although not the standard one. The equation obeyed by
$G(\rho) = \rho^{-\alpha}H(\rho)$ is the standard normal form of the
confluent Heun equation (with one indicial exponent vanishing at
both $0$ and $1$).

The scheme \cref{eq:generalscheme} lists the asymptotic
forms of the local solutions about the singular points.
In particular, the indicial exponents $\lambda_1 = 0$ and
$\lambda_1'=2\sigma-s$ correspond to Frobenius-type solutions
\cite{wasow1987asymptotic}
$\widehat{\Phi}_\ell^\mathrm{up} \sim 1$
and $\widehat{\Phi}_\ell^\mathrm{down}
\sim (\rho - 1)^{2\sigma - s}$ as $\rho \rightarrow 1$.
The transformation \eqref{eq:CHtrans} then gives
$\widehat{\psimode}_\ell^\mathrm{down} \sim
\rho^{-s}(\rho - 1)^{2\sigma}e^{-\sigma\rho_*}
\sim e^{-2\sigma} e^{\sigma\rho_*}$,
matching the first equation in \eqref{eq:BarackSpecialSolutions},
upon rescaling $\widehat{\psimode}_\ell^\mathrm{down}\rightarrow
e^{2\sigma}\widehat{\psimode}_\ell^\mathrm{down}$
and passing from dimensionless to physical variables. For the
point at $\infty$ we seek solutions $\widehat{\Phi}_\ell
\sim \rho^{-\lambda} e^{\kappa \rho}
\sum_{n=0}^\infty c_n \rho^{-n}$.
According to Erd\'{e}lyi \cite{erdelyi1956asymptotic}, such
expansions are due to Thom\'{e}. From \cref{eq:generalscheme}
we have $\kappa_\infty = 0, \lambda_\infty = 0$ and
$\kappa_\infty' = 2\sigma, \lambda_\infty' = -2\sigma - 2s$.
The local solutions are then $\widehat{\Phi}^\mathrm{out}_\ell \sim 1$ 
and $\widehat{\Phi}^\mathrm{inc}_\ell \sim
\rho^{2\sigma+2s} e^{2\sigma \rho}$ as $\rho\rightarrow\infty$.
Then $\widehat{\psimode}_\ell^\mathrm{out} \sim e^{-\sigma\rho_*}$
from \eqref{eq:CHtrans}, confirming the last equation in
\eqref{eq:BarackSpecialSolutions}. For $\sigma\neq 0$ write
\begin{align}\label{eq:Thome2}
\widehat{\Phi}^\mathrm{out}_\ell(\sigma,\rho)
\sim \sum_{n=0}^\infty \frac{c_n(\sigma)}{\rho^n}.
\end{align}
Then with $L \equiv \ell(\ell+1)-s(s+1) = (\ell-s)(\ell+s+1)$,
$S = (s+1)(2s+1)$, and $c_0=1$, we first find $c_1 =
\frac{1}{2}L/\sigma$. Thereupon, we find the three-term recursion
\begin{align}
2\sigma (n+2)c_{n+2}
= \big[L - (n+1)(n + 2s + 2)\big] c_{n+1}
+ [n(n+3s+2)+S]c_{n}.
\end{align}

When $\sigma =0$, \eqref{eq:BardeenPress-Phihat} becomes
\begin{align}\label{eq:BardeenPress-Phihat_0}
\begin{split}
\frac{d^2\widehat{\Phi}_\ell}{d \rho^2}
& +\Big[-\frac{3s+1}{\rho} + \frac{1+s}{\rho-1}\Big]
\frac{d \widehat{\Phi}_\ell}{d\rho}
\\
& + \Big[
\frac{(s+1)(2s+1)}{\rho^2}
-\frac{(s+1)^2 + \ell (\ell + 1)}{\rho (\rho-1)}
\Big]\widehat{\Phi}_\ell = 0.
\end{split}
\end{align}
Now the point at $\infty$ is a regular singular point, and
\eqref{eq:BardeenPress-Phihat} is an incarnation of the
Gauss-hypergeometric equation, although it is not quite in
the standard normal form. We seek Frobenius-type solutions 
$\widehat{\Phi}_\ell = \rho^{-\lambda} \sum_{n=0}^\infty a_n\rho^{-n}$.
The possible $\lambda$ are $\ell-s$ and $-1-\ell-s$. We expect
$\ell \geq |s|$, and so choose $\lambda = \ell-s$.
Standard calculations yield
\begin{align}\label{eq:zerosig-aRecur}
a_0=1,\quad
(n+1)(2\ell + n + 2)a_{n+1} = (\ell + n + 1)(\ell + n + s + 1)a_n
\text{ for } n \geq 1.
\end{align}

\section{Evaluation of the FDRK}
\label{app:KernelEvaluation}
This appendix considers numerical evaluation of
$\widehat{\omega}_{\ell}(iy,\rho_b)$ for $y\geq 0$. Through
\cref{eq:TLPfromRBC}, it also pertains to evaluation of
teleportation kernels. Evaluation for negative $y$-values
follows from the respective even- and odd-parity of
$\mathrm{Re}\widehat{\omega}_\ell(iy,\rho_b)$ and
$\mathrm{Im}\widehat{\omega}_\ell(iy,\rho_b)$.
The parity of these profiles follows from
\cref{eq:Riccatirho}. For fixed $y_\mathrm{cut} > 0$
evaluation falls into three cases:
(i) $0 < y_\mathrm{cut} < y$,
(ii) $0 < y \leq y_\mathrm{cut}$, and
(iii) $y=0$.
Case (i) relies on a two-path integration scheme over a
straight segment in the complex $z$-plane, where
$z=\sigma\rho$. Case (ii) relies on a three-path integration
scheme in the complex $\rho$-plane.
Respectively, these cases are described
further in
\cref{appsubsec:two-path,appsubsec:three-path}.
Case (iii) relies on the explicit formula
\begin{align}
\widehat{\omega}_\ell(0,\rho) = \rho \partial_\rho
\log(\sum_{n=0}^\infty a_n \rho^{-(\ell-s+n)}) 
\rightarrow -\ell + s \text{ as } \rho\rightarrow\infty,
\end{align}
defined in terms of recursion \cref{eq:zerosig-aRecur}
and suitable truncation. \Cref{subsec:EvalKernels} examines
the accuracy of schemes (i) and (ii).

\subsection{Two-component path integration for case (i)}
\label{appsubsec:two-path}
Write the ODE \eqref{eq:BardeenPress-Phihat} as
\begin{align}\label{eq:ODEzee}
\frac{d^2\widehat{\Phi}_\ell}{dz^2}
+ \frac{P(z/\sigma,\sigma)}{\sigma} \frac{d\widehat{\Phi}_\ell}{dz}
+ \frac{Q(z/\sigma)}{\sigma^2}\widehat{\Phi}_\ell = 0,
\end{align}
where $z=\sigma\rho$. Likewise, write \eqref{eq:Riccatirho} as
\begin{align}\label{eq:Riccatizee}
\frac{d\widehat{\omega}_\ell}{dz}
+ \frac{\widehat{\omega}_\ell^2}{z}
+ \Big[\frac{P(z/\sigma,\sigma)}{\sigma}
- \frac{1}{z}\Big]\widehat{\omega}_\ell
+ z\frac{Q(z/\sigma)}{\sigma^2} = 0.
\end{align}
With the asymptotic expansion \eqref{eq:Thome2} in
$\rho$, we get an asymptotic expansion
$\widehat{\Phi}_\ell^\mathrm{out} \sim \sum_{n=0}^\infty
d_n(\sigma) z^{-n}$ for the
outgoing solution viewed as a function of $z=\sigma\rho$,
using $d_n(\sigma) = c_n(\sigma) \sigma^n$. Suitable
truncation of this $z$-expansion is then used to
approximate the initial data
$\{\widehat{\Phi}_{\ell}^\mathrm{out}|_1,
d\widehat{\Phi}_{\ell}^\mathrm{out}/dz|_1\}$,
where $|_1$ indicates evaluation at a
point $z_1 = \mathsf{scale}_1 +\mathrm{i}y \rho_b$, where
$\mathsf{scale}_1 \gg \mathsf{scale}_2 > 0$ is a large
positive number. With this initial data, we then integrate
\eqref{eq:ODEzee} to an intermediate $z$-value 
$z_2 = \mathsf{scale}_2 +\mathrm{i}y \rho_b$ over the blue
path shown in the left panel of
\cref{fig:TwoAndThreeComponentPaths}. As described
in \cite{Lau:2004jn}, over this integration the second solution
$\widehat{\Phi}_\ell^\mathrm{inc}$ (initially present due
to inexactness in the initial data) is exponentially
suppressed. The result is initial data $\widehat{\omega}_\ell|_2$
for the Riccati equation good to machine accuracy. Then
\eqref{eq:Riccatizee} is integrated over the red path,
terminating at $z_b = \mathrm{i}y \rho_b$ with the kernel
value. For $0 < y \ll 1$ this method breaks down, since the
terms in \eqref{eq:Riccatizee} become singular, although
their combination is regular.
\begin{figure}[!htbp]
\begin{center}
\includegraphics[width=7.75cm]{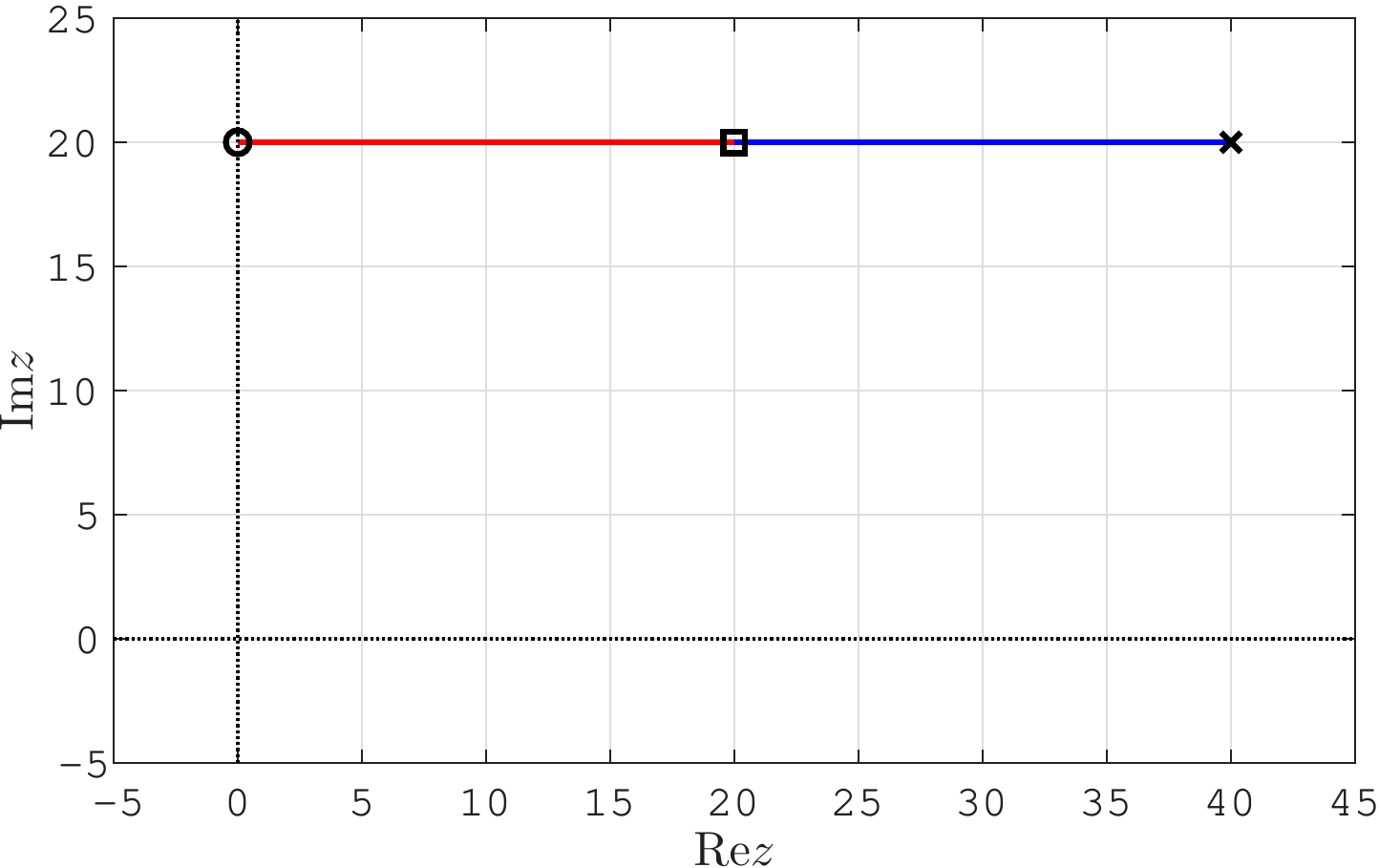}
\hspace{0.5cm}
\includegraphics[width=4.05cm]{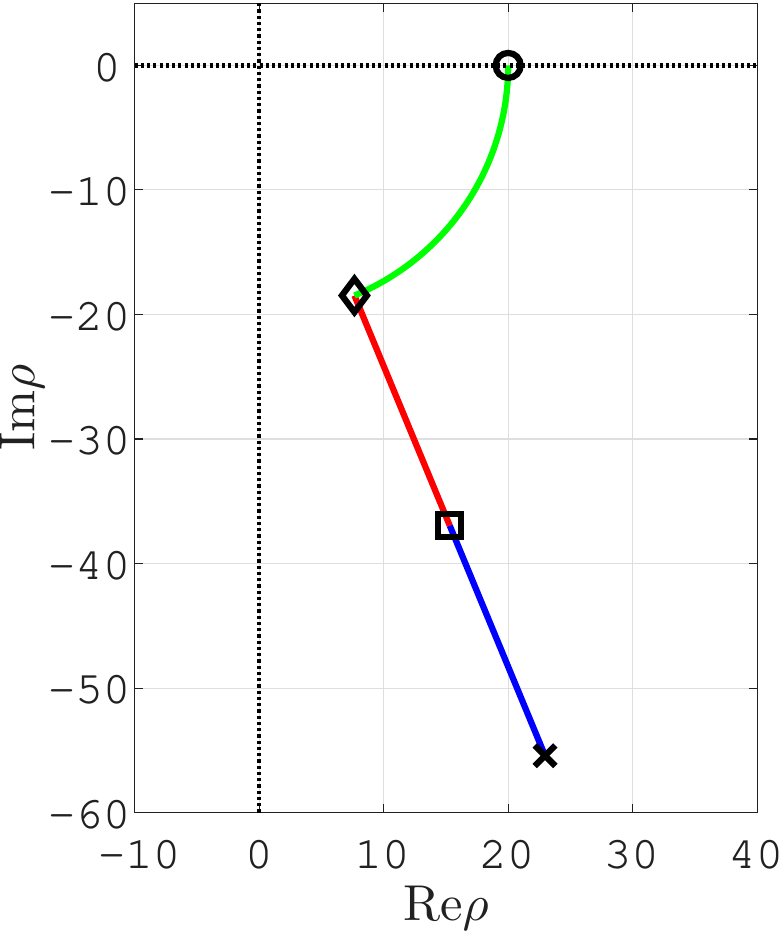}
\caption{Left panel: Representative
two-component path in $z$-plane.
Right panel: Representative
three-component path in $\rho$-plane.
\label{fig:TwoAndThreeComponentPaths}}
\end{center}
\end{figure}
\subsection{Three-component path integration for case (ii)}
\label{appsubsec:three-path}
We complexify $\rho = \chi e^{\mathrm{i\psi}}$, with $\psi$
initially chosen such that $\mathrm{Re}(\sigma\rho) =
\mathrm{Re}(\mathrm{i}y\rho) > 0$. For the sake of definiteness,
with $\sigma = \mathrm{i}y$ for $y >0$, take $\psi = -\frac{1}{4}\pi$.
Next, for positive numbers $\mathsf{scale}_1 > \mathsf{scale}_2 >
\mathsf{scale}_3$, define $\rho_1 = \mathsf{scale}_1e^{\mathrm{-i\pi/4}}$,
$\rho_2 = \mathsf{scale}_2e^{\mathrm{-i\pi/4}}$,
and $\rho_3 = \rho_b e^{\mathrm{-i\pi/4}}$.

With the asymptotic expansion \eqref{eq:Thome2} in
$\rho$, we get initial data approximating
$\{\widehat{\Phi}_{\ell}^\mathrm{out}|_1,
d\widehat{\Phi}_{\ell}^\mathrm{out}/d\chi|_1\}$,
where $|_1$ indicates evaluation at $\rho_1$.
Next, we integrate the ODE
\begin{align}
\frac{d^2\widehat{\Phi}_\ell}{d\chi^2}
+ e^{\mathrm{i}\psi}P(\chi e^{\mathrm{i\psi}},\sigma)
\frac{d\widehat{\Phi}_\ell}{d\chi}
+ e^{\mathrm{2i}\psi}Q(\chi e^{\mathrm{i\psi}})\widehat{\Phi}_\ell = 0
\end{align}
from $\chi = \mathsf{scale}_1$ to $\chi = \mathsf{scale}_2$ 
(the blue path in right panel of \cref{fig:TwoAndThreeComponentPaths}).
The Riccati equation \eqref{eq:Riccatirho} can then be
viewed as an equation in $\chi$; we integrate this $\chi$-Riccati
equation from $\chi = \mathsf{scale}_2$ to $\chi = \rho_b$ (red path).
The last step is consider \eqref{eq:Riccatirho} as a
Riccati equation in $\psi$, integrating from
$\psi = \psi_1$ to $\psi = 0$ (green path).

\subsection{Accuracy in kernel evaluation}
\label{subsec:EvalKernels}
The two-path scheme (i) breaks down for
$0 < y \ll y_\mathrm{cut}$ and the three-path (ii) scheme breaks
down for $y \gg y_\mathrm{cut}$; see \cite{Lau:2004jn}. The
experiment here is to compare the two schemes on a neighborhood
of $y_\mathrm{cut}$. Precisely, we consider
evaluation of $\widehat{\omega}_2(iy,60)$,
with $y_\mathrm{cut} = 0.005$ and the comparison taking place
over a logarithmically spaced partition of the 
interval $[10^{-5},10^1]$. Both integration schemes use
$\mathsf{scale}_1=1000$, $\mathsf{scale}_2 = 100$, and the fifth-order
Runge-Kutta-Fehlberg method, with 131072 steps on straight segments
and 2048 steps on the final arc path for the three-component scheme.
\Cref{fig:2pathvs3path} shows the
relative errors in $\mathrm{Re}\widehat{\omega}_2(iy,60)$
and $\mathrm{Im}\widehat{\omega}_2(iy,60)$, obtained via
comparison of the two methods. The plot shows that both the
two-component and three-component schemes are accurate on a
window $10^{-3} \leq y \leq 10^{-2}$ which contains
$y_\mathrm{cut} = 0.005$ (shown as a dotted line).
To the left of the dotted line, evaluation via the
three-component scheme serves as the denominator in relative
errors, and to the right evaluation via the two-component scheme.
On the overlap region $[10^{-3},10^{-2}]$ the agreement between
the two schemes is better than
$10^{-12}$, with computations performed in double precision.
We have also implemented these schemes in quadruple
precision with correspondingly better comparison. In
conclusion, with our two methods (and series evaluation at
the origin) accurate evaluation of $\widehat{\omega}_2(iy,60)$
is available for all $y\geq 0$. 
\begin{figure}[!htbp]
\centering
\includegraphics[width=10cm]{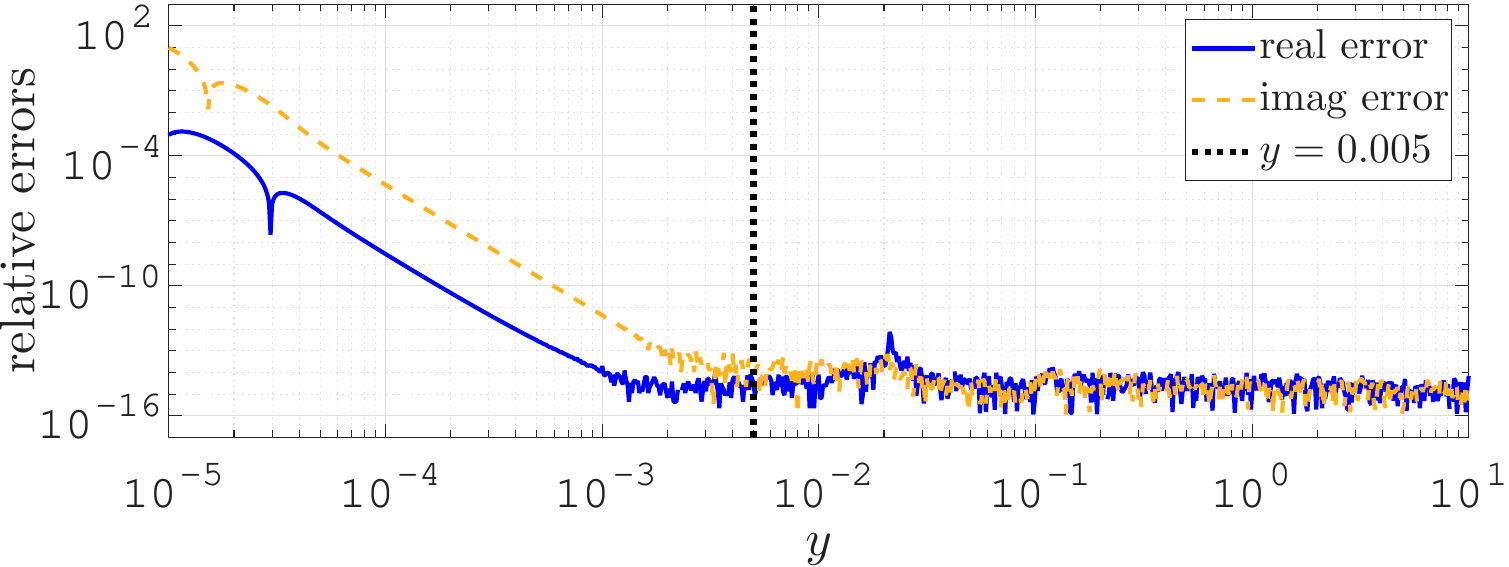}
\caption{
Relative difference in
$\widehat{\omega}_{2}(iy;60)$ computed via comparison of
the two-component and three-component path integration schemes.
The comparison indicates that both schemes are accurate on the
$[10^{-3},10^{-2}] \ni y_\mathrm{cut} = 0.005$
(dotted line).
}
\label{fig:2pathvs3path}
\end{figure}

\section{Error bounds}
\subsection{\texorpdfstring{$L^2$-norm error bounds for approximate kernels}{
L2-norm error bounds for approximate kernels}}\label{sub:l2_norm_bounds}
This section studies the error in our approximations for boundary conditions
and signal teleportation. Above we have adopted the dimensionless Laplace
frequency $\sigma = 2Ms$, where now $s$ is the physical Laplace frequency.
This meaning for $s$ clashes with its use as the spin in the Bardeen-Press
equation. In this section $s$ always means the Laplace variable. 

Suppose that we have an ``exact'' (causal) kernel $B(t)$ and an associated
convolution
\begin{align}
(B*\psimode)(t) = \int_0^t B(t-t')\psimode(t')dt'.
\end{align}
A function $f(t)$ is causal if $f(t)=f(t)\theta(t)$, where $\theta(t)$ is
the Heaviside step function; equivalently, $f(t)=0$ for $t<0$. Here, $B(t)$
may be a boundary, teleportation, or asymptotic evaluation kernel. In all
cases, a numerically constructed approximation $A(t)$ replaces the exact kernel.
We wish to bound the error incurred when $A(t)$ replaces $B(t)$.
Reference~\cite{Benedict:2012kw} derived error estimates for this replacement
over the infinite time window $[0,\infty)$. We now generalize these estimates
to the finite-time interval $[0,T]$. 

\begin{theorem}[$L_2$ convolution error bounds]\label{lem:finite_horizon_est}
For $T>0$ let $A(t)$ and $B(t)$ be causal kernels, and $\psimode(t)$ a
causal input. Assume the Laplace transforms $\widehat{A}(s)$ and
$\widehat{B}(s)$ exist for $s\in \mathrm{i}\mathbb{R}$, with well-defined
convolutions $(A*\psimode)(t)$ and $(B*\psimode)(t)$ for
$t\ge 0$.
\vskip 5pt

\noindent
(i) (Relative error) If $\widehat{B}(s)\neq 0$ for all
$s\in \mathrm{i}\mathbb{R}$, then
\begin{align}\label{eq:l2_estimate_rel_T}
\|A*\psimode-B*\psimode\|_{L_2(0,T)} \leq
\sup_{s\in\mathrm{i}\mathbb{R}} 
\frac{|\widehat{A}(s)-\widehat{B}(s)|}{|\widehat{B}(s)|}\;
\|B*\psimode\|_{L_2(0,T)}.
\end{align}
\vskip 5pt

\noindent
(ii) (Absolute error) Without the nonvanishing condition
on $\widehat{B}$, one has
\begin{equation}\label{eq:l2_estimate_abs_T}
\|A*\psimode-B*\psimode\|_{L_2(0,T)} \leq
(2\pi)^{-1/2} \sup_{s\in\mathrm{i}\mathbb{R}}
|\widehat{A}(s)-\widehat{B}(s)|\; \|\psimode\|_{L_2(0,T)}.
\end{equation}
\end{theorem}
\begin{proof}
Define the restriction operator $P_T$ by
\begin{align}
(P_T f)(t) = f(t) \mathbf{1}_{[0,T]}(t),
\end{align}
where $\mathbf{1}_{[0,T]}(t)$ is the indicator function.
For a causal kernel $K$ and causal input $f$,
\begin{align}\label{eq:PT_identity}
P_T(K*f) = P_T\bigl(K*(P_T f)\bigr),
\end{align}
since $(K*f)(t)=\int_0^t K(t-t')f(t')dt'$
for $0\le t\le T$ only samples $f(t)$ on $[0,t]\subseteq [0,T]$.

Let $D(t)=A(t)-B(t)$. With $L_2$ monotonicity under restriction
and \cref{eq:PT_identity},
\begin{align}
\|D*\psimode\|_{L_2(0,T)}
= \|P_T(D*(P_T\psimode))\|_{L_2(0,\infty)}
\le \|D*(P_T\psimode)\|_{L_2(0,\infty)}.
\end{align}
Now apply the $L_2(0,\infty)$ absolute estimate from
Eq.~(A3) of Ref.~\cite{Benedict:2012kw} with input $P_T\psimode$:
\begin{align}
\|D*(P_T\psimode)\|_{L_2(0,\infty)}
\le (2\pi)^{-1/2}\Big(\sup_{s\in \mathrm{i}\mathbb{R}}|\widehat{D}(s)|
\Big)\;\|P_T\psimode\|_{L_2(0,\infty)}.
\end{align}
Since $\widehat{D}(s)=\widehat{A}(s)-\widehat{B}(s)$ and
$\|P_T\psimode\|_{L_2(0,\infty)}=\|\psimode\|_{L_2(0,T)}$,
this result is \eqref{eq:l2_estimate_abs_T}.

Now assume $\widehat{B}(s)\neq 0$ for all $s\in \mathrm{i}\mathbb{R}$
and define 
\begin{align}
\widehat{E}(s) = \frac{\widehat{A}(s)-\widehat{B}(s)}{\widehat{B}(s)},
\qquad s\in \mathrm{i}\mathbb{R}.
\end{align}
Starting with the frequency-domain identity
$(\widehat{A}-\widehat{B})\widehat{\psimode}=\widehat{E}
(\widehat{B}\widehat{\psimode})$,
one obtains
\begin{align}
(A-B)*\psimode = E*(B*\psimode),
\end{align}
via the convolution theorem.
Set $f=B*\Psi$, and use \cref{eq:PT_identity} along with the
same restriction argument as above, thereby finding
\begin{align}
\begin{split}
\|(A-B)*\Psi\|_{L_2(0,T)} & = \|P_T(E*f)\|_{L_2(0,\infty)} \\
& = \|P_T(E*(P_T f))\|_{L_2(0,\infty)} \\
& \le \|E*(P_T f)\|_{L_2(0,\infty)}.
\end{split}
\end{align}
Finally, apply the $L_2(0,\infty)$ relative estimate from
\cite{Benedict:2012kw} get the bound
\begin{align}
\|E*(P_T f)\|_{L_2(0,\infty)} & \le
\Big(\sup_{s\in \mathrm{i}\mathbb{R}}
|\widehat{E}(s)|\Big)\|P_T f\|_{L_2(0,\infty)}
= \Big(\sup_{s\in \mathrm{i}\mathbb{R}}|\widehat{E}(s)|\Big)
\|f\|_{L_2(0,T)}.
\end{align}
Recalling that $f=B*\Psi$ and
$\widehat{E}(s)=(\widehat{A}(s)-\widehat{B}(s))/\widehat{B}(s)$,
we establish \cref{eq:l2_estimate_rel_T}.
\end{proof}

\subsection{Max-norm error bounds for approximate kernels}
\label{sec:max_norm_error}
For causal input $\psimode(t)$ the convolution error at time $t\ge 0$ is
$e(t) = \int_0^t D(t')\psimode(t-t')dt'$,
where, for a reason given below, we have changed variables to have
$\psimode$ with the shifted argument. Consider a final time $T>0$
and choose a window size $\Delta \in (0,T]$. For a fixed $t\in[0,T]$,
we split the integral into recent-history and remote-history
portions,
\begin{align}
\label{eq:split_error_norm}
e(t) = \underbrace{\int_0^{\min(t,\Delta)} D(t')
\psimode(t-t')dt'}_{\text{recent-history}}
     + \underbrace{\int_{\Delta}^{\max(t,\Delta)} D(t')
\psimode(t-t')dt'}_{\text{remote-history}}.
\end{align}
Notice that only the recent-history portion contributes to $e(t)$ for
$t\leq \Delta$. We interpret $\Delta$ as an ``effective memory''
time for the convolution. Choosing a $\Delta > 0$, corresponds to
our intuition that if the kernel is decaying
with time, then the error at late times should depend more strongly
on the signal at late times. 
\begin{theorem}[Pointwise-maximum error bound]
\label{theorem:max_error}
Let $T>0$, $\Delta\in(0,T]$, with $t\geq \Delta$.
For $\psimode\in L_\infty(0,T)$ and
any causal $D\in L_1(0,T)$, the pointwise error obeys
\begin{equation}
\label{eq:max_error}
|e(t)|
\le
\|D\|_{L_1(0,\Delta)}\,\|\psimode\|_{L_\infty(t-\Delta,t)}
+
\|D\|_{L_1(\Delta,t)}\,\|\psimode\|_{L_\infty(0,t-\Delta)}.
\end{equation}
\end{theorem}

\begin{proof}
We begin from \cref{eq:split_error_norm}. For the first term,
note that $t-t'\in[t-\Delta,t]$ for $t'\in[0,\Delta]$. For the
second, note that $t-t'\in[0,t-\Delta]$ for $t'\in[\Delta,t]$.
These observations and standard estimates complete the proof.
\end{proof}

\begin{remark}\label{remark2}
The first term on the right-hand side of \eqref{eq:max_error} depends only on
the size of $\psimode$ in the window $[t-\Delta,t]$, while the second
term couples the early-time amplitude of $\psimode$ to late-time size 
$\|D\|_{L_1(\Delta,t)}$ of the kernel error. If $D(t')$ decays in $t'$,
$\Delta$ can be chosen so that $\|D\|_{L_1(\Delta,t)}$ is sufficiently
small for all $t$ of interest. For a decaying kernel error
early-time amplitudes are suppressed by the smallness of $\|D\|_{L_1(\Delta,t)}$.
\end{remark}

As a practical matter, we note that while the maximum pointwise (temporal)
error is controlled by the $L_1$-norm of the temporal kernel error
$D(t)=A(t)-B(t)$, the sum-of-poles approximation directly controls the
kernel's relative error on the inversion contour~\cref{eq:SumofPoles_error}.
Nevertheless, $\|D\|_{L^1(0,T)}$ can be computed (or bounded) a posteriori,
and it generally decreases as the approximation
\cref{eq:SumofPoles_error} is tightened.

\subsection{Application of error bounds}



Our kernel computation introduces two sources of error in the
frequency-domain radiation, teleportation, and AWE kernels: numerical
evaluation of the exact Laplace-domain profiles along the inversion
contour, typically $\sigma=iy$, and compression of the resulting data
into a sum-of-poles form for efficient time-domain use. The bounds
developed above apply to both steps, although in practice we mainly use
them to quantify the compression error, since the first step can be
driven close to machine precision by accurate ODE integration along
numerically stable paths together with sufficiently refined $y$-sampling
(cf.~\cref{app:KernelEvaluation}). Concretely, once a rational
approximation $\widehat{\xi}_\ell(iy) \simeq
\widehat{\omega}_\ell(iy,\rho_b)$ for, say, an FDRK has been obtained,
we evaluate the pointwise difference
$D(Y_\alpha) =
\widehat{\xi}_\ell(iY_\alpha)-\widehat{\omega}_\ell(iY_\alpha,\rho_b)$
on a dense auxiliary grid
$\{Y_\alpha\}_{\alpha=0}^{N_Y}$, with $N_Y$ much larger than for the
compression grid$\{y_j\}_{j=0}^{2QN_{\mathrm{int}}}$; see
\cref{subsec:CompressionAlg}. The quantities $D$ and $E$ entering the
bounds are then obtained directly from these sampled differences.

\section{Tables}
\label{app:numtables}
This appendix collects the kernels used for the experiments
in \cref{sec:NumExp}. \Cref{tab:ROBC30table} determines the
approximation $\xi_{2}(t,30)$, and \Cref{tab:ROBC60table}
the approximation $\xi_{2}(t,60)$. These tables have been
computed in quad precision and achieve the tolerance
$\varepsilon = 10^{-13}$. \Cref{tab:TLP30to60table} determines
$\xi^T_{2}(t,30,60)$. This table has been computed in quad precision
and achieves the tolerance $\varepsilon = 10^{-13}$. Finally,
\Cref{tab:TLP60to1e15table} determines $\xi^T_{2}(t,60,10^{15})$.
This table has been computed in double precision, and we believe the
approximation achieves the tolerance $\varepsilon = 10^{-12}$.
\begin{table}[!h]
\centering
\caption{ROBC table for $s = -2$, $\ell = 2$, $\rho = 30$, $\varepsilon = 10^{-13}$
\label{tab:ROBC30table}
}
\hrule
\vskip 5pt
\input{tables/FdCqRho0030L002D18_bp2.table}
\vskip 5pt
\hrule
\end{table}

\begin{table}[!h]
\centering
\caption{ROBC table for $s = -2$, $\ell = 2$, $\rho = 60$, $\varepsilon = 10^{-13}$
\label{tab:ROBC60table}
}
\hrule
\vskip 5pt
\input{tables/FdCqRho0060L002D18_bp2.table}
\vskip 5pt
\hrule
\end{table}

\begin{table}[!h]
\centering
\caption{Teleportation table for $s = -2$, $\ell = 2$, $\rho_1 = 30$, $\rho_2 = 60$, $\varepsilon = 10^{-13}$
\label{tab:TLP30to60table}
}
\hrule
\vskip 5pt
\input{tables/FdCqRho0030to0060L002D20_bp2.table}
\vskip 5pt
\hrule
\end{table}
\begin{table}[!h]
\centering
\caption{Teleportation table for $s = -2$, $\ell = 2$, $\rho_1 = 60$,
                             $\rho_2 = 10^{15}$, $\varepsilon = 10^{-12}$
\label{tab:TLP60to1e15table}
}
\hrule
\vskip 5pt
\input{tables/FdCqRho0060to1E15L002D15_bp2.table}
\vskip 5pt
\hrule
\end{table}

\bibliographystyle{siamadvanced} 

\begin{thebibliography}{10}

\bibitem{SpECwebsite}
{\em The {S}pectral {E}instein {C}ode}.
\newblock \url{http://www.black-holes.org/SpEC.html}.

\bibitem{Berti:2025hly}
{\sc J.~Abedi et~al.}, {\em {Black hole spectroscopy: from theory to
  experiment}},  (2025), \url{https://arxiv.org/abs/2505.23895}.

\bibitem{AbramowitzStegun}
{\sc M.~Abramowitz and I.~A. Stegun}, {\em Handbook of Mathematical
  Functions with Formulas, Graphs, and Mathematical Tables}, 
  Applied Mathematics Series, vol.~55, 10th printing,
  December 1972, with corrections.

\bibitem{LISAConsortiumWaveformWorkingGroup:2023arg}
{\sc N.~Afshordi et~al.}, {\em {Waveform modelling for the Laser Interferometer
  Space Antenna}}, Living Rev. Rel., 28 (2025), 9,
  \url{https://arxiv.org/abs/2311.01300}.

\bibitem{alpert2000rapid}
{\sc B.~Alpert, L.~Greengard, and T.~Hagstrom}, {\em Rapid evaluation of
  nonreflecting boundary kernels for time-domain wave propagation}, SIAM
  Journal on Numerical Analysis, 37 (2000), pp.~1138--1164.

\bibitem{alpert2002nonreflecting}
{\sc B.~Alpert, L.~Greengard, and T.~Hagstrom}, {\em Nonreflecting boundary
  conditions for the time-dependent wave equation},
  J. Comput. Phys., 180 (2002), pp.~270--296.

\bibitem{Ansorg:2016ztf}
{\sc M.~Ansorg and R.~Panosso~Macedo}, {\em {Spectral decomposition of
  black-hole perturbations on hyperboloidal slices}}, Phys. Rev. D, 93 (2016),
  124016, \url{https://arxiv.org/abs/1604.02261}.

\bibitem{APPELO20094200}
{\sc D.~Appel\"{o} and T.~Colonius}, {\em A high-order super-grid-scale
  absorbing layer and its application to linear hyperbolic systems}, 
  J. Comput. Phys., 228 (2009), pp.~4200--4217.

\bibitem{Aretakis:2023ast}
{\sc S.~Aretakis, G.~Khanna, and S.~Sabharwal}, {\em {An observational
  signature for extremal black holes}}, Gen. Rel. Grav., 57 (2025), 160,
  \url{https://arxiv.org/abs/2307.03963}.

\bibitem{Barack:2017oir}
{\sc L.~Barack and P.~Giudice}, {\em {Time-domain metric reconstruction for
  self-force applications}}, Phys. Rev. D, 95 (2017), 104033,
  \url{https://arxiv.org/abs/1702.04204}.

\bibitem{Bardeen:1973xb}
{\sc J.~M. Bardeen and W.~H. Press}, {\em {Radiation fields in the
  Schwarzschild background}}, J. Math. Phys., 14 (1973), pp.~7--19.

\bibitem{https://doi.org/10.1002/cpa.3160330603}
{\sc A.~Bayliss and E.~Turkel}, {\em Radiation boundary conditions for
  wave-like equations}, Commun. Pure Appl. Math., 33
  (1980), pp.~707--725,
  \url{https://arxiv.org/abs/https://onlinelibrary.wiley.com/doi/pdf/10.1002/cpa.3160330603}.

\bibitem{Benedict:2012kw}
{\sc A.~G. Benedict, S.~E. Field, and S.~R. Lau}, {\em {Fast evaluation of
  asymptotic waveforms from gravitational perturbations}}, Class. Quant. Grav.,
  30 (2013), 055015, \url{https://arxiv.org/abs/1210.1565}.

\bibitem{Berti:2009kk}
{\sc E.~Berti, V.~Cardoso, and A.~O. Starinets}, {\em {Quasinormal modes of
  black holes and black branes}}, Class. Quant. Grav., 26 (2009), 163001,
  \url{https://arxiv.org/abs/0905.2975}.

\bibitem{Bishoyi:2024lqm}
{\sc S.~D. Bishoyi, S.~Sabharwal, and G.~Khanna}, {\em {Numerical Evidence for
  Non-Axisymmetric Gravitational {\textquotedblleft}Hair{\textquotedblright}
  for Extremal Kerr Black Hole Spacetimes with Hyperboloidal Foliations}}, Gen.
  Rel. Grav., 57 (2025), 46, \url{https://arxiv.org/abs/2407.06926}.

\bibitem{BIZZOZERO2017118}
{\sc D.~Bizzozero, J.~Ellison, K.~Heinemann, and S.~Lau}, {\em Rapid evaluation
  of two-dimensional retarded time integrals}, J. Comput. Appl. Math., 324
  (2017), pp.~118--141.

\bibitem{Brito:2015oca}
{\sc R.~Brito, V.~Cardoso, and P.~Pani}, {\em {Superradiance}: {New Frontiers
  in Black Hole Physics}}, Lect. Notes Phys., 906 (2015), pp.~1--237,
  \url{https://arxiv.org/abs/1501.06570}.

\bibitem{Buchman_2024}
{\sc L.~T. Buchman, M.~D. Duez, M.~Morales, M.~A. Scheel, T.~M. Kostersitz,
  A.~M. Evans, and K.~Mitman}, {\em Numerical relativity multimodal waveforms
  using absorbing boundary conditions}, Class. Quant. Grav, 41
  (2024), 175011.

\bibitem{Buchman_2006}
{\sc L.~T. Buchman and O.~C.~A. Sarbach}, {\em Towards absorbing outer
  boundaries in general relativity}, Class. Quant. Grav, 23 (2006),
  pp.~6709--6744.

\bibitem{Buchman_2007}
{\sc L.~T. Buchman and O.~C.~A. Sarbach}, {\em Improved outer boundary
  conditions for einstein’s field equations}, Class. Quant. Grav,
  24 (2007), pp.~S307--S326.

\bibitem{Burko:2020wzq}
{\sc L.~M. Burko, G.~Khanna, and S.~Sabharwal}, {\em {Scalar and gravitational
  hair for extreme Kerr black holes}}, Phys. Rev. D, 103 (2021), L021502,
  \url{https://arxiv.org/abs/2005.07294}.

\bibitem{Cardoso:2021wlq}
{\sc V.~Cardoso, K.~Destounis, F.~Duque, R.~P. Macedo, and A.~Maselli}, {\em
  {Black holes in galaxies: Environmental impact on gravitational-wave
  generation and propagation}}, Phys. Rev. D, 105 (2022), L061501,
  \url{https://arxiv.org/abs/2109.00005}.

\bibitem{10.1098/rspa.1975.0066}
{\sc S.~Chandrasekhar}, {\em On the equations governing the perturbations of
  the Schwarzschild black hole}, Proc. Roy. Soc. Lond. A, 343
  (1975), pp.~289--298,
  \url{https://arxiv.org/abs/https://royalsocietypublishing.org/rspa/article-pdf/343/1634/289/61331/rspa.1975.0066.pdf}.

\bibitem{Chandrasekhar:BHs}
{\sc S.~Chandrasekhar}, {\em The Mathematical Theory of Black Holes}, Clarendon
  Press, 1998.

\bibitem{PhysRevD.100.104025}
{\sc K.~Csuk\'as, I.~R\'acz, and G.~Z. T\'oth}, {\em Numerical investigation of
  the dynamics of linear spin $s$ fields on a Kerr background: Late-time tails
  of spin $s=\pm 1, \pm 2$
  fields}, Phys. Rev. D, 100 (2019), 104025.

\bibitem{Dolan:2012yt}
{\sc S.~R. Dolan}, {\em {Superradiant instabilities of rotating black holes in
  the time domain}}, Phys. Rev. D, 87 (2013), 124026,
  \url{https://arxiv.org/abs/1212.1477}.

\bibitem{donninger2011proof}
{\sc R.~Donninger, W.~Schlag, and A.~Soffer}, {\em A proof of Price's law on
  Schwarzschild black hole manifolds for all angular momenta}, Adv. Math.,
  226 (2011), pp.~484--540.

\bibitem{erdelyi1956asymptotic}
{\sc A.~Erd{\'e}lyi}, {\em Asymptotic Expansions}, Dover Books on Mathematics,
  Dover Publications, New York, 1956.
\newblock Unabridged and unaltered republication of Technical report 3,
  prepared under contract Nonr-220(11) for the Office of Naval Research, 1955.

\bibitem{Field:2009kk}
{\sc S.~E. Field, J.~S. Hesthaven, and S.~R. Lau}, {\em {Discontinuous Galerkin
  method for computing gravitational waveforms from extreme mass ratio
  binaries}}, Class. Quant. Grav., 26 (2009), 165010,
  \url{https://arxiv.org/abs/0902.1287}.

\bibitem{Field:2014cka}
{\sc S.~E. Field and S.~R. Lau}, {\em {Fast evaluation of far-field signals for
  time-domain wave propagation}}, J. Sci. Comput., 64 (2015), pp.~647--669,
  \url{https://arxiv.org/abs/1409.5893}.

\bibitem{Gelles:2025gxi}
{\sc Z.~Gelles and F.~Pretorius}, {\em {Accumulation of charge on an extremal
  black hole{\textquoteright}s event horizon}}, Phys. Rev. D, 112 (2025),
  064003, \url{https://arxiv.org/abs/2503.04881}.

\bibitem{Greengard_2014}
{\sc L.~Greengard, T.~Hagstrom, and S.~Jiang}, {\em The solution of the scalar
  wave equation in the exterior of a sphere}, J. Comput. Phys.,
  274 (2014), pp.~191--207.

\bibitem{HAGSTROM1998403}
{\sc T.~Hagstrom and S.~Hariharan}, {\em A formulation of asymptotic and exact
  boundary conditions using local operators}, Appl. Numer. Math., 27
  (1998), pp.~403--416.
\newblock Special Issue on Absorbing Boundary Conditions.

\bibitem{PhysRevD.62.044029}
{\sc S.~A. Hughes}, {\em Computing radiation from Kerr black holes:
  Generalization of the Sasaki-Nakamura equation}, Phys. Rev. D, 62 (2000),
  044029.

\bibitem{Hughes:2019zmt}
{\sc S.~A. Hughes, A.~Apte, G.~Khanna, and H.~Lim}, {\em {Learning about black
  hole binaries from their ringdown spectra}}, Phys. Rev. Lett., 123 (2019),
  161101, \url{https://arxiv.org/abs/1901.05900}.

\bibitem{Islam:2024vro}
{\sc T.~Islam, G.~Faggioli, G.~Khanna, S.~E. Field, M.~van~de Meent, and
  A.~Buonanno}, {\em {Phenomenology and origin of late-time tails in eccentric
  binary black hole mergers}}, Phys. Rev. D, 112 (2025), 024061,
  \url{https://arxiv.org/abs/2407.04682}.

\bibitem{Lau:2004as}
{\sc S.~R. Lau}, {\em {Rapid evaluation of radiation boundary kernels for
  time-domain wave propagation on black holes: Implementation and numerical
  tests}}, Class. Quant. Grav., 21 (2004), pp.~4147--4192.

\bibitem{Lau:2004jn}
{\sc S.~R. Lau}, {\em {Rapid evaluation of radiation boundary kernels for time
  domain wave propagation on blackholes: theory and numerical methods}},
  J. Comput. Phys., 199 (2004), pp.~376--422,
  \url{https://arxiv.org/abs/gr-qc/0401001}.

\bibitem{Lau:2005ti}
{\sc S.~R. Lau}, {\em {Analytic structure of radiation boundary kernels for
  blackhole perturbations}}, J. Math. Phys., 46 (2005), 102503,
  \url{https://arxiv.org/abs/gr-qc/0507140}.

\bibitem{Leaver:1985ax}
{\sc E.~W. Leaver}, {\em {An Analytic representation for the quasi-normal modes
  of Kerr black holes}}, Proc. Roy. Soc. Lond. A, 402 (1985), pp.~285--298.

\bibitem{LiTafloveBackman2005}
{\sc X.~Li, A.~Taflove, and V.~Backman}, {\em Modified fdtd near-to-far-field
  transformation for improved backscattering calculation of strongly
  forward-scattering objects}, IEEE Antennas Wirel. Propag. Lett., 4
  (2005), pp.~35--38.

\bibitem{Lim:2022veo}
{\sc H.~Lim, S.~A. Hughes, and G.~Khanna}, {\em {Measuring quasinormal mode
  amplitudes with misaligned binary black hole ringdowns}}, Phys. Rev. D, 105
  (2022), 124030, \url{https://arxiv.org/abs/2204.06007}.

\bibitem{Long:2021ufh}
{\sc O.~Long and L.~Barack}, {\em {Time-domain metric reconstruction for
  hyperbolic scattering}}, Phys. Rev. D, 104 (2021), 024014,
  \url{https://arxiv.org/abs/2105.05630}.

\bibitem{Ma:2023qjn}
{\sc S.~Ma et~al.}, {\em {Fully relativistic three-dimensional
  Cauchy-characteristic matching for physical degrees of freedom}}, Phys. Rev.
  D, 109 (2024), 124027, \url{https://arxiv.org/abs/2308.10361}.

\bibitem{Ma:2024hzq}
{\sc S.~Ma, M.~A. Scheel, J.~Moxon, K.~C. Nelli, N.~Deppe, L.~E. Kidder,
  W.~Throwe, and N.~L. Vu}, {\em {Merging black holes with
  Cauchy-characteristic matching: Computation of late-time tails}}, Phys. Rev.
  D, 112 (2025), 024003, \url{https://arxiv.org/abs/2412.06906}.

\bibitem{Muller2011}
{\sc J.~Muller, G.~Parent, G.~Jeandel, and D.~Lacroix}, {\em Finite-difference
  time-domain and near-field-to-far-field transformation in the spectral
  domain}, J. Opt. Soc. Am. A, 28 (2011), pp.~868--878.

\bibitem{PanossoMacedo:2019npm}
{\sc R.~Panosso~Macedo}, {\em {Hyperboloidal framework for the Kerr
  spacetime}}, Class. Quant. Grav., 37 (2020), 065019,
  \url{https://arxiv.org/abs/1910.13452}.

\bibitem{Price:1972pw}
{\sc R.~H. Price}, {\em {Nonspherical Perturbations of Relativistic
  Gravitational Collapse. II. Integer-Spin, Zero-Rest-Mass Fields}}, Phys. Rev.
  D, 5 (1972), pp.~2439--2454.

\bibitem{Rink:2024swg}
{\sc K.~Rink, R.~Bachhar, T.~Islam, N.~E.~M. Rifat, K.~Gonzalez-Quesada, S.~E.
  Field, G.~Khanna, S.~A. Hughes, and V.~Varma}, {\em {Gravitational wave
  surrogate model for spinning, intermediate mass ratio binaries based on
  perturbation theory and numerical relativity}}, Phys. Rev. D, 110 (2024),
  124069, \url{https://arxiv.org/abs/2407.18319}.

\bibitem{Rinne:2008vn}
{\sc O.~Rinne, L.~T. Buchman, M.~A. Scheel, and H.~P. Pfeiffer}, {\em
  {Implementation of higher-order absorbing boundary conditions for the
  Einstein equations}}, Class. Quant. Grav., 26 (2009), 075009,
  \url{https://arxiv.org/abs/0811.3593}.

\bibitem{slavyanov2000special}
{\sc S.~Y. Slavyanov and W.~Lay}, {\em Special Functions: A Unified Theory
  Based on Singularities}, Oxford Mathematical Monographs, Oxford University
  Press, Oxford, 2000.

\bibitem{Stein:2012ffl}
{\sc L.~C. Stein}, {\em {Probes of strong-field gravity}}, PhD thesis, MIT,
  2012.

\bibitem{Sundararajan:2007jg}
{\sc P.~A. Sundararajan, G.~Khanna, and S.~A. Hughes}, {\em {Towards adiabatic
  waveforms for inspiral into Kerr black holes. I. A New model of the source
  for the time domain perturbation equation}}, Phys. Rev. D, 76 (2007),
  104005, \url{https://arxiv.org/abs/gr-qc/0703028}.

\bibitem{TafloveHagness2005}
{\sc A.~Taflove and S.~C. Hagness}, {\em Computational Electrodynamics: The
  Finite-Difference Time-Domain Method}, Artech House, Norwood, MA, 3rd ed.,
  2005.

\bibitem{Teukolsky:1973ApJ}
{\sc S.~A. {Teukolsky}}, {\em {Perturbations of a Rotating Black Hole. I.
  Fundamental Equations for Gravitational, Electromagnetic, and Neutrino-Field
  Perturbations}}, Astrophys. J., 185 (1973), pp.~635--648.

\bibitem{Teukolsky:1974yv}
{\sc S.~A. Teukolsky and W.~H. Press}, {\em {Perturbations of a rotating black
  hole. III. Interaction of the hole with gravitational and electromagnetic
  radiation}}, Astrophys. J., 193 (1974), pp.~443--461.

\bibitem{Thorne:1980ru}
{\sc K.~S. Thorne}, {\em {Multipole Expansions of Gravitational Radiation}},
  Rev. Mod. Phys., 52 (1980), pp.~299--339.

\bibitem{Tokita1972}
{\sc T.~Tokita}, {\em Exponential decay of solutions for the wave equation in
  the exterior domain with spherical boundary}, Math. Kyoto Univ, 12
  (1972), pp.~413--430.

\bibitem{Vishal:2023fye}
{\sc M.~Vishal, S.~E. Field, K.~Rink, S.~Gottlieb, and G.~Khanna}, {\em {Toward
  exponentially-convergent simulations of extreme-mass-ratio inspirals: A
  time-domain solver for the scalar Teukolsky equation with singular source
  terms}}, Phys. Rev. D, 110 (2024), 104009,
  \url{https://arxiv.org/abs/2307.01349}.

\bibitem{Wardell:2021fyy}
{\sc B.~Wardell, A.~Pound, N.~Warburton, J.~Miller, L.~Durkan, and A.~Le~Tiec},
  {\em {Gravitational Waveforms for Compact Binaries from Second-Order
  Self-Force Theory}}, Phys. Rev. Lett., 130 (2023), 241402,
  \url{https://arxiv.org/abs/2112.12265}.

\bibitem{wasow1987asymptotic}
{\sc W.~Wasow}, {\em Asymptotic Expansions for Ordinary Differential
  Equations}, Dover Books on Mathematics, Dover Publications, Inc., New York,
  1987.
\newblock Reprint of the 1965 edition.

\bibitem{wilcox1959initial}
{\sc C.~H. Wilcox}, {\em The initial-boundary value problem for the wave
  equation in an exterior domain with spherical boundary},
  Notices of the AMS, 6 (1959), pp.~869--870.
  \newblock Abstract No. 564--20.

\bibitem{xu2013bootstrap}
{\sc K.~Xu and S.~Jiang}, {\em A bootstrap method for sum-of-poles
  approximations}, J. Sci. Comput., 55 (2013), pp.~16--39.

\bibitem{Zenginoglu:2010cq}
{\sc A.~Zenginoglu}, {\em {Hyperboloidal layers for hyperbolic equations on
  unbounded domains}}, J. Comput. Phys., 230 (2011), pp.~2286--2302,
  \url{https://arxiv.org/abs/1008.3809}.

\bibitem{Zenginoglu:2011zz}
{\sc A.~Zenginoglu and G.~Khanna}, {\em {Null infinity waveforms from
  extreme-mass-ratio inspirals in Kerr spacetime}}, Phys. Rev. X, 1 (2011),
  021017, \url{https://arxiv.org/abs/1108.1816}.

\bibitem{Zenginoglu:2008uc}
{\sc A.~Zenginoglu, D.~Nunez, and S.~Husa}, {\em {Gravitational perturbations
  of Schwarzschild spacetime at null infinity and the hyperboloidal initial
  value problem}}, Class. Quant. Grav., 26 (2009), 035009,
  \url{https://arxiv.org/abs/0810.1929}.

\end{thebibliography}

\end{document}